\documentclass[letterpaper,11pt]{article}

\usepackage[margin=1in]{geometry}
\usepackage[nospace,noadjust]{cite}
\usepackage[title]{appendix}

\newcommand{\etaKL}{\eta_{\rm KL}}

\newcommand{\etaTV}{\eta_{\rm TV}}

\def\dperp{\perp\!\!\!\perp}

\usepackage[
            CJKbookmarks=true,
            bookmarksnumbered=true,
            bookmarksopen=true,
            colorlinks=true,
            citecolor=red,
            linkcolor=blue,
            anchorcolor=red,
            urlcolor=blue
            ]{hyperref}

\usepackage{mdframed}
\definecolor{light-gray}{gray}{.92}

\usepackage{microtype}
\usepackage{xspace,mathtools,soul}

\usepackage{graphicx}
\usepackage{amssymb}
\usepackage{amsmath}%
\usepackage{amsfonts}%
\usepackage{amssymb}%
\usepackage{amsthm}
\usepackage{graphicx}
\usepackage[table,xcdraw,usenames,dvipsnames]{xcolor}
\usepackage{subfigure}
\usepackage{psfrag}
\usepackage{enumerate}

\newtheorem{thm}{Theorem}
\newtheorem{lem}{Lemma}

\newtheorem{prop}{Proposition}

\newtheorem{cor}{Corollary}

\theoremstyle{definition}

\newtheorem{defn}{Definition}

\newtheorem{remark}{Remark}

\newcommand{\calX}{\mathcal{X}}

\newcommand{\calD}{\mathcal{D}}

\newcommand{\calY}{\mathcal{Y}}
\newcommand{\calL}{\mathcal{L}}

\newcommand{\calF}{\mathcal{F}}
\newcommand{\calG}{\mathcal{G}}

\newcommand{\calP}{\mathcal{P}}
\newcommand{\calN}{\mathcal{N}}

\newcommand{\floor}[1]{\lfloor #1 \rfloor}

\newcommand{\reals}{\mathbb{R}}

\newcommand{\integers}{\mathbb{Z}}
\newcommand{\Expect}{\mathbb{E}}

\newcommand{\prob}[1]{\mathbb{P}\left[#1\right]}

\newcommand{\TV}{d_{\rm TV}}

\newcommand{\diff}{{\rm d}}
\newcommand{\eg}{e.g.\xspace}
\newcommand{\ie}{i.e.\xspace}

\renewcommand{\tilde}{\widetilde}

\newcommand{\Xh}{\hat{X}}
\newcommand{\xhw}{\hat{x}_w}

\newcommand{\Reals}{\mathbb{R}}

\newcommand{\ones}{\mathbf{1}}

\newcommand{\defined}{\triangleq}

\newcommand{\ExpVal}[2]{\mathbb{E}_{#1}\left[ #2 \right]}

\newcommand{\EE}[1]{\ExpVal{}{#1}}

\newcommand{\mmse}{\mathsf{mmse}}
\newcommand{\lse}{\mathsf{lmmse}}
\newcommand{\dks}{d_{\mathrm{KS}}}
\newcommand{\Xt}{\widetilde{X}}

\newcommand{\etakl}{\eta_{\mathrm{KL}}}

\newcommand{\etabb}{\bar{\eta}}

\newcommand\independent{\protect\mathpalette{\protect\independenT}{\perp}}
\def\independenT#1#2{\mathrel{\rlap{$#1#2$}\mkern2mu{#1#2}}}

\newcommand{\pth}[1]{\left( #1 \right)}

\newcommand{\fracd}[2]{\frac{\diff #1}{\diff #2}}

\usepackage{prettyref}
\newrefformat{eq}{(\ref{#1})}
\newrefformat{thm}{Theorem~\ref{#1}}
\newrefformat{th}{Theorem~\ref{#1}}
\newrefformat{chap}{Chapter~\ref{#1}}
\newrefformat{sec}{Section~\ref{#1}}
\newrefformat{algo}{Algorithm~\ref{#1}}
\newrefformat{fig}{Fig.~\ref{#1}}
\newrefformat{tab}{Table~\ref{#1}}
\newrefformat{rmk}{Remark~\ref{#1}}
\newrefformat{clm}{Claim~\ref{#1}}
\newrefformat{def}{Definition~\ref{#1}}
\newrefformat{cor}{Corollary~\ref{#1}}
\newrefformat{lmm}{Lemma~\ref{#1}}
\newrefformat{prop}{Proposition~\ref{#1}}
\newrefformat{pr}{Proposition~\ref{#1}}
\newrefformat{app}{Appendix~\ref{#1}}
\newrefformat{apx}{Appendix~\ref{#1}}
\newrefformat{ex}{Example~\ref{#1}}
\newrefformat{exer}{Exercise~\ref{#1}}
\newrefformat{soln}{Solution~\ref{#1}}

\renewcommand{\Xt}{\tilde{X}}
\newcommand{\tX}{\tilde{X}}

\newcommand{\diverge}{\to\infty}

\renewcommand{\Pr}{\mathbb{P}}
\newcommand{\cf}{\varphi}

\long\def\apxonly#1{}

\title{
Strong Data Processing Inequalities for Input Constrained Additive Noise Channels
}
\author{Flavio P.~Calmon \and Yury Polyanskiy \and Yihong Wu
    \thanks{F.~P.~Calmon is with the Harvard John A.~Paulson School of Engineering and Applied Sciences. E-mail:
    \texttt{flavio@seas.harvard.edu}. Y.~Polyanskiy is with the Department of Electrical
  Engineering and Computer Science, MIT, Cambridge, MA, 02139, USA.
  E-mail: \texttt{yp@mit.edu}. Y.~Wu is with the Department of Statistics and Data Science,
 Yale University, New Haven, CT, 06511, USA. E-mail: \texttt{yihong.wu@yale.edu}.
 This work is supported  in part by the National Science Foundation (NSF) CAREER awards under Grant CCF-12-53205 and CCF-1651588, the NSF Grant IIS-1447879, CCF-1423088, CCF-1527105,  and by the Center for Science of Information (CSoI), an NSF Science and Technology Center, under Grant CCF-09-39370. This paper was presented in part at the 2015 IEEE International Symposium on Information Theory.
}  
}

\begin{document}
\maketitle

\begin{abstract}
This paper quantifies the intuitive observation that adding noise reduces available information by means of non-linear
strong data processing inequalities. Consider the random variables $W\to X\to Y$ forming a Markov chain, where $Y=X+Z$
with $X$ and $Z$ real-valued, independent and $X$ bounded in $L_p$-norm. 
It is shown that $I(W;Y) \le F_I(I(W;X))$ with $F_I(t)<t$ whenever $t>0$, if and only if $Z$ has a density whose support is not disjoint from any translate of itself.

A related question is to characterize for what couplings $(W,X)$ the mutual information $I(W;Y)$ is close to maximum
possible.  To that end we show that in order to saturate the channel, i.e. for $I(W;Y)$ to approach capacity, it is
mandatory that $I(W;X)\to\infty$ (under suitable conditions on the channel). A key ingredient for this result is a
deconvolution lemma which shows that post-convolution total variation distance bounds the pre-convolution
Kolmogorov-Smirnov distance. 

Explicit bounds are provided for the special case of the additive Gaussian noise channel with quadratic cost constraint.
These bounds are shown to be order-optimal. For this case simplified proofs are provided
leveraging Gaussian-specific tools such as the connection between information and estimation (I-MMSE) and Talagrand's information-transportation inequality.
\end{abstract}

\newpage
\tableofcontents


\section{Introduction}

Strong data-processing inequalities (SDPIs) quantify the decrease of mutual information under the action of a noisy channel.
Such inequalities have apparently been first discovered by Ahlswede and G\'acs in a landmark
paper~\cite{ahlswede1976spreading}. Among the work predating~\cite{ahlswede1976spreading} and extending it we
mention~\cite{dobrushin_central_1956,sarmanov1962maximum,CIR93}.
Notable connections include topics
ranging from existence and uniqueness of Gibbs measures and log-Sobolev inequalities to performance limits of noisy
circuits. We refer the reader to the introduction in~\cite{polyanskiy_dissipation_2014} and the recent
monographs~\cite{raginsky_strong_2014,raginsky2013concentration} for more detailed discussions of applications and extensions.

For a fixed channel $P_{Y|X}:\calX\to \calY$, let $P_{Y|X}\circ P$ be
the distribution on $\calY$ induced by the push-forward of the distribution $P$.
One approach to strong data processing seeks to find the contraction coefficients 
\begin{equation}
  \label{eq:eta}
    \eta_f \defined
    \sup_{P,Q:P\neq Q}\frac{D_f\left(P_{Y|X}\circ P\|P_{Y|X}\circ
    Q\right)}{D_f(P\|Q)}\,,
\end{equation}
where the $D_f(P\|Q) \triangleq \Expect_Q[f(\frac{dP}{dQ})]$ is an $f$-divergence of Csisz\'ar~\cite{IC67}.
When the divergence $D_f$ is the KL-divergence and total
variation,\footnote{The total variation between two distributions $P$ and $Q$ is
$\TV(P,Q)\defined \sup_E|P[E]-Q[E]|$.} we denote the
coefficient $\eta_f$ as $\etakl$ and $\etaTV$, respectively. 

For discrete channels,~\cite{ahlswede1976spreading} showed 
that strict contraction for KL-divergence is equivalent to strict contraction in terms of total variation ($\etakl<1\Leftrightarrow \etaTV<1$), and $\etakl<1$ if an only if the bipartite graph describing the channel, determined by the edges $$\left\{(x,y)\mid P_{Y|X}(y|x)>0 \right\},$$ is connected.
 Having $\etakl<1$ implies reduction in the usual data-processing inequality for
mutual information~\cite[Exercise III.2.12]{CK},~\cite{anantharam2013maximal}:
\begin{equation}
\label{eq:LinearDPIs}
\forall\ W\to X\to Y: I(W;Y) \le \etakl \cdot I(W;X)\,.
\end{equation}
We refer to inequalities of the form \eqref{eq:LinearDPIs} as \textit{linear} SDPIs.

When $P_{Y|X}$ is an additive white Gaussian noise channel, i.e. $Y=X+Z$ with $Z\sim \calN(0,1)$,
it has been shown~\cite{polyanskiy_dissipation_2014} that restricting the maximization in~\eqref{eq:eta} to
distributions with a bounded second moment (or any moment) still leads to no-contraction, giving $\etakl = \etaTV=1$ for
AWGN. Nevertheless, the contraction does indeed take place, except not multiplicatively.
The region \[\left\{\left(\TV(P,Q
  ),\TV(P*P_Z,Q*P_Z)\right):\mathbb{E}_{(P+Q)/2}[X^2] \leq \gamma \right\}\,, \]
has been explicitly determined in~\cite{polyanskiy_dissipation_2014},  where $*$ denotes convolution. The boundary of this region, dubbed the \textit{Dobrushin curve} of the channel, turned out to be
strictly bounded away from the diagonal (identity). In other words, except for the trivial case where $\TV(P,Q)=0$, total variation
decreases by a non-trivial amount in Gaussian channels. 

Unfortunately, the similar region for KL-divergence turns out to be trivial, so that no improvement in the inequality
$$ D(P_X*P_Z \| Q_Z*P_Z) \le D(P_X\|Q_X) $$
is possible (given the knowledge of the right-hand side and moment constraints on $P_X$ and $Q_X$).
In~\cite{polyanskiy_dissipation_2014}, in order to study how mutual information dissipates on a chain of Gaussian links, this problem was resolved by a rather lengthy workaround which entails first reducing 
questions regarding the mutual information to those about the total variation and then converting back. 

A more direct approach, in the spirit of the joint-range idea of Harremo\"es and Vajda~\cite{harremoes2011pairs}, is to find (or bound) the \emph{best possible data-processing function} $F_I$ defined as follows.
\begin{defn}
  \label{def:FI}
  For a fixed channel $P_{Y|X}$ and a convex set $\calP$ of distributions on $\calX$ we define
  \begin{equation}
        F_I(t,P_{Y|X},\calP)\defined \sup \left\{ I(W;Y)\colon I(W;X)\leq t,
    W\rightarrow X\rightarrow Y, P_X \in \calP
  \right\},
  \end{equation}
where the supremum is over all joint distributions $P_{W,X}$ with $P_X\in\calP$. When the channel is clear from the context, we abbreviate $F_I(t,P_{Y|X})$ as $F_I(t)$.
\end{defn}

For brevity we denote $F_I(t,\gamma)$ the function corresponding to the special case of the AWGN channel and quadratic
constraint. Namely, for $Y_\gamma=\sqrt{\gamma}X+Z$, where $Z\sim \calN(0,1)$ is independent of $X$, we define
\begin{equation}\label{eq:figamma}
  F_I(t,\gamma)\defined \sup \left\{ I(W;Y_\gamma)\colon I(W;X)\leq t,
    W\rightarrow X\rightarrow Y_\gamma, \EE{X^2} \le 1
  \right\}.
\end{equation}

The significance of the function $F_I$ is that it gives the optimal input-independent strong data processing
inequalities. It is instructive to compare definition of $F_I$ with two related quantities considered previously in the
literature.  Witsenhausen and Wyner~\cite{WW75} defined
\begin{equation}
  F_T(P_{XY}, h) = \inf H(Y|W), 
  \label{eq:FH}
\end{equation}
with the infimum taken over all joint distributions satisfying
$$ W\to X\to Y, H(X|W)=h, \Pr[X=x, Y=y] = P_{XY}(x,y)\,.$$
Clearly, by a simple reparametrization $h=H(X)-t$, this function would correspond to $H(Y)-F_I(t)$ if $F_I(t)$ were defined
with restriction to a given input distribution $P_X$. The $P_X$-independent version of \prettyref{eq:FH} has also been studied by
Witsenhausen~\cite{HW74}:
$$ f_T(P_{Y|X}, h) = \inf H(Y|W), $$
with the infimum taken over all
$$ W\to X \to Y, H(X|W)=h, \Pr[Y=y|X=x] = P_{Y|X}(y|x)\,.$$
This quantity plays a role in a generalization of Mrs. Gerber's lemma and satisfies a convenient tensorization
property:
$$ f_T((P_{Y|X})^n, nh) = n f_T(P_{Y|X}, h)\,. $$
There is no one-to-one correspondence between $f_T(P_{Y|X}, h)$ and $F_I(t)$ and in fact, alas, $F_I(t)$ does not
satisfy any (known to us) tensorization property.

\subsection{Overview of results}
A priori, the only bounds we can state on $F_I$ are consequences of capacity and the data processing inequality:
\begin{equation}
  F_I(t,P_{Y|X})\leq \min\left\{ t,C(P_{Y|X}, \calP) \right\},
  \label{eq:trivial}
\end{equation}
where $C(P_{Y|X},\calP)\defined \sup_{P_X\in\calP} I(X;Y)$. For the Gaussian-quadratic case, 
\begin{equation}
  F_I(t,\gamma)\leq \min\left\{ t,C(\gamma) \right\},
\end{equation}
where $C(\gamma)\defined\frac{1}{2}\log(1+\gamma)$ is the Gaussian channel capacity\footnote{All logarithms are in base $e$.}.

In this work we show that generally the trivial bound~\eqref{eq:trivial} is not tight at any point. Namely, we prove
that 
\begin{align} F_I(t) &\le t - g_d(t), \label{eq:gd} \\
  F_I(t) &\le C - g_h(t) \label{eq:gh}
\end{align}
and both functions $g_d$ and $g_h$ are strictly positive for all $t>0$. We call these two results \textit{diagonal} and
\textit{horizontal} bounds respectively. See Fig.~\ref{fig:fireg} for an illustration. 

\begin{figure}[tb]
  \centering
  \psfrag{x}[cc]{\small $I(W;X)$}
  \psfrag{y}[cc]{\small $I(W;Y) $ }
  \psfrag{a}[cc][c]{\small $~g_d$ }
  \psfrag{b}[cc][c]{\small $~g_h$ } 
  \psfrag{f}[Bc]{$F_I$ }
  \includegraphics[width=0.6\textwidth]{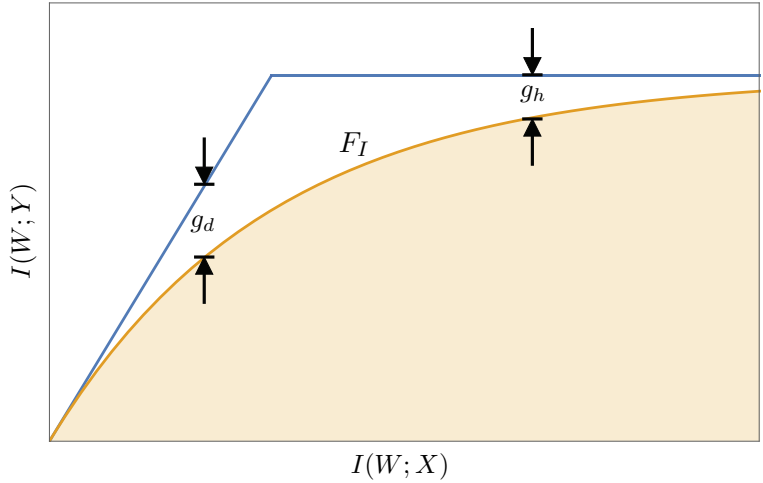}
  \vspace{-0.05in}
\caption{The strong data processing function $F_I$ and gaps $g_d$ and $g_h$ to the trivial data processing bound \eqref{eq:trivial}.}
  \label{fig:fireg}
  \vspace{-0.1in}
\end{figure}

For the Gaussian-quadratic case we show explicitly that our estimates are asymptotically
sharp. For example, Theorem \ref{thm:diagonal} (Gaussian diagonal bound) shows the lower-bound portion of
\begin{equation}
    g_d(t,\gamma) =  e^{-\frac{\gamma}{t}\log\frac{1}{t}+ \Theta(\log \frac{1}{t})}.
 \label{eq:FuIneq}
\end{equation}
An application of~\eqref{eq:FuIneq} allows, via a repeated application of~\eqref{eq:gd}, to infer that the mutual information between the input $X_0$ and the output
$Y_n$ of a chain of $n$ energy-constrained Gaussian relays converges to zero $I(X_0; Y_n)\to0$. In
fact,~\eqref{eq:FuIneq} recovers the optimal convergence rate of $\Theta(\frac{\log \log n}{\log n})$ first reported
in~\cite[Theorem 1]{polyanskiy_dissipation_2014}. 

We then generalize the diagonal bound to non-Gaussian noise and arbitrary moment constraint (Theorem
\ref{thm:DiagGeneral}) by an additional quantization argument. 
It is worth noting that mutual information does not always strictly contract. Consider the following simple example: Let $Z$ be uniformly distributed over $[0,1]$ and $W=X$ is Bernoulli, then $I(W;X+Z)=I(W;X)=H(X)$ since $X$ can be decoded perfectly from $X+Z$. Surprisingly, this turns out to be the only situation for non-contraction of mutual information to occur, as the following characterization (\prettyref{cor:strict}) shows:
for strict contraction of mutual information it is \emph{necessary and sufficient}
that the noise $Z$ cannot be perfectly distinguished from a translate of itself (i.e. $\TV(P_Z, P_{Z+x}) \neq 1$).

Going to the horizontal bound, we show (for the Gaussian-quadratic case) that $F_I(t,\gamma)$ approaches
$C(\gamma)$ no faster than double-exponentially in $t$ as $t\rightarrow \infty$. Namely, in Theorem \ref{thm:capBound} and Remark
\ref{remark:cap}, we prove that $g_h(t)$ satisfies
\begin{equation}
  \label{eq:bound_gh}
  e^{-c_1(\gamma)e^{4t}} \leq g_h(t)\leq e^{-c_2(\gamma)e^t+\log 4(1+\gamma)},
\end{equation}
where $c_1(\gamma)$ and $c_2(\gamma)$ are strictly positive functions of $\gamma$.

Generalization of the horizontal bound to arbitrary noise distribution (Theorem~\ref{thm:generalHoriz}) proceeds along a
similar route. In the process, we derive a deconvolution estimate that bounds the Kolmogorov-Smirnov distance ($L_\infty$ norm between CDFs)
in terms of the total variation between convolutions with
noise. Namely, \prettyref{cor:deconv-phiz} shows that
for a noise $Z$ with bounded density and non-vanishing characteristic function we have
$$ \dks(P,Q) \leq f(\TV(P*P_Z, Q*P_Z)) $$
for some continuous increasing function $f(\cdot)$ with $f(0)=0$.

The final result (\prettyref{thm:infinite}) addresses the question of bounding $F_I$-curve for non-scalar channel
$Y=X+Z$. Somewhat surprisingly, we show that for the infinite-dimensional Gaussian case the trivial bound \eqref{eq:trivial} on
the $F_I$-curve is exact.

\subsection{Organization and notation}

The rest of the paper is organized as follows. Section \ref{sec:FIproperties} introduces properties of the $F_I$-curve,
together with a few examples for discrete channels. 

Sections \ref{sec:diagonal} and~\ref{sec:GeneralDiagonal} present a (diagonal) lower bound for
$g_d(t)$ in the Gaussian and generall setting respectively. Section \ref{sec:mmse} shows that any $X$ for which
close-to-optimal (in MMSE sense) linear estimator of $Y=X+Z$ exists, must necessarily be close to Gaussian in the sense
of Kolmogorov-Smirnov distance. These results are then used in Section
\ref{sec:horizontal} to prove a (Gaussian horizontal) lower bound on $g_h(t)$. 

Section \ref{sec:deconvTV} introduces a deconvolution result that
connects KS-distance with TV-divergence. This result is then applied in Section \ref{sec:GeneralHorizontal} to derive a
general horizontal bound for $F_I$ curve for a wide range of additive noise channels. 

Finally, in  Section \ref{sec:infinite} we
consider the infinite-dimensional discrete Gaussian channel, and show that in this case there exists no non-trivial
strong data processing inequality for mutual information. In the appendix, we present a shorter proof of the
key step in the Gaussian horizontal bound (namely, Lemma \ref{lem:capKS}) using Talagrand's inequality
\cite{talagrand_transportation_1996}.

\paragraph{Notations}
For any distribution $P$ on $\reals$, let $F_P(x) = P((-\infty,x])$ denote its cumulative distribution function (CDF). For any random variable $X$, denote its distribution and CDF by $P_X$ and $F_X$, respectively.
For any sequences $\{a_n\}$ and $\{b_n\}$ of positive numbers, we write $a_n\gtrsim b_n$ or $b_n\lesssim a_n$ when $a_n\geq cb_n$ for some absolute constant $c>0$. 

\section{Examples and properties of the $F_I$-curves}

\label{sec:FIproperties}

In this section we discuss properties of the $F_I$-curve, and present a few examples for discrete channels.

\begin{prop}[Properties of the $F_I$-curve]~~~~~~~~~~~~~~~~~~
\begin{enumerate}
	\item $F_I$ is an increasing function such that $0 \leq F_I(t) \leq t$ with $F_I(0)=0$.
	\item $t \mapsto \frac{F_I(t)}{t}$ is decreasing. Consequently, $F_I$ is subadditive and $F_I'(0) = \sup_{t > 0} \frac{F_I(t)}{t}$.
	\item Value of $F_I(t)$ is unchanged if $W$ is restricted to an alphabet of size $|\calX|+1$. 
	Upper concave envelope of $F_I(t)$ equals upper concave envelope of a set of pairs $(I(W;X), I(W;Y))$ achieved
	by restricting $W$ to alphabet $\calX$.
\end{enumerate}
	\label{prop:FIprop}
\end{prop}
\begin{proof}
The first part follows directly from the definition, the non-negativity and the data processing inequality of mutual information. For the second part, fix $P_{Y|X}$ and let $P_{WX}$ achieve the pair $(I(W;X),I(W;Y))$. Then by choosing $P'_{WX}=\lambda P_{WX} +(1-\lambda)P_WP_X$, the pair $(\lambda I(W;X),\lambda I(W;Y))$ is also achievable. It follows directly that $t \mapsto F_I(t)/t$ is decreasing.

	Claim 3 follows by noticing that for a fixed distribution $P_X$, any pair $(H(X|W), H(Y|W))$ can be attained by
	$W$ with a given restriction on the alphabet, see~\cite[Theorem 2.3]{WW75} or \cite[Appendix C]{el2011network}. Similarly, concave envelope of $F_I(t)$
	can be found by taking convex closure of extremal points $(H(X) - H(X|W), H(Y) - H(Y|W))$, which can be attained
	by $W$ with alphabet $|\calX|$, see paragraph after~\cite[Theorem 2.3]{WW75}.
\end{proof}

\apxonly{
\textbf{TODO: } Add other remarks on computing $F_I$-curve, namely: bound on cardinality of $W$; reduction to uniform $W$
with unbounded cardinality; expression for concavified $F_I$-curve as $I(U_1; Y|U_2)$; convexity of the lower-boundary of
$F_I$-region; $P_{WX}$ at boundary of $F_I$ $\iff$ at boundary of $KL$-divergence.
}

We present next a few examples of the $F_I(t)$-curve for discrete channels:

\begin{enumerate}
\item \textit{Erasure channel} is defined as $P_{Y|X}:\calX\to\calX\cup\{?\}$ with $y=x$ or $?$ with probabilities
$1-\alpha$ and $\alpha$, respectively. In this case we have for any $W-X-Y$ a convenient identity, cf.~\cite{VW08}:
		$$ I(W;Y)=(1-\alpha)I(W;X)\,,$$
	and consequently, the $F_I$-curve is
	\begin{equation}
		F_I(t) = (1-\alpha)\left(t \wedge \log |\cal X| \right)
	\end{equation}
	and is achieved by taking $W=X$.
\item \textit{Binary symmetric channel} BSC($\delta$) is defined as $P_{Y|X}: \{0,1\}\to\{0,1\}$ with $Y=X+Z$, $Z \sim
\mathrm{Ber}(\delta)$. Here the optimal coupling is $X=W+Z'$ with $Z'
\dperp W \sim \mathrm{Ber}(1/2)$ and varying bias of $Z'$. This is formally proved in the next Proposition and illustrated in Fig. .
\end{enumerate}

\begin{figure}[tb]
  \centering
  \psfrag{x}[cc]{$t$}
  \psfrag{y}[cc][b][1][180]{$F_I(t)$ }
  \psfrag{C}[cc][c]{\small $\delta=0.1,0.2,0.3,0.4$ }
  \includegraphics[width=0.6\textwidth]{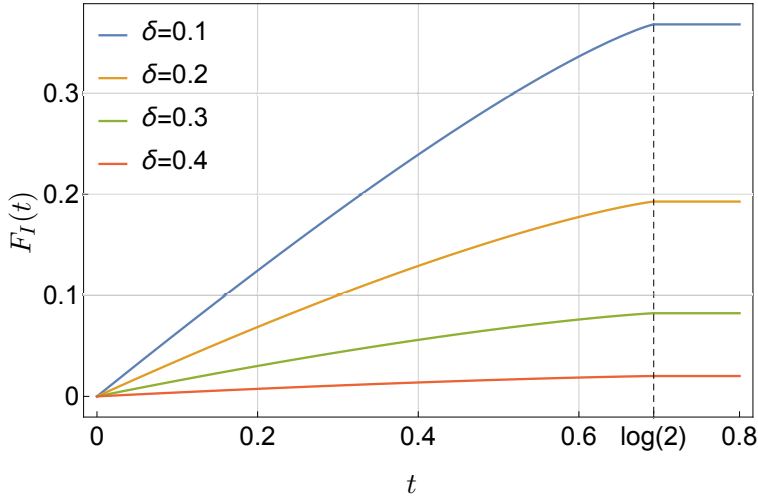}
  \vspace{-0.05in}
\caption{$F_I$-curve \eqref{eq:FIBSC} for $BSC(\delta)$.}
  \label{fig:fireg}
  \vspace{-0.1in}
\end{figure}

\begin{prop}
\label{prop:FIBSC}
The $F_I$-curve of the BSC($\delta$) is given by
\begin{equation}
F_I(t) = \log 2 - h_b\left(\delta*h_b^{-1}(|\log 2-t|^+)\right),
\label{eq:FIBSC}
\end{equation}
where $p * q = p(1-q)+q(1-p)$, $h_b(y)\defined -y\log y-(1-y)\log(1-y)$ is the binary entropy function and $h_b^{-1}:
[0,\log 2] \to [0,{1\over2}]$ is its functional inverse. 
\end{prop}
\begin{proof}
\label{sec:BSC_gerber}
First, it is clear that
\begin{equation}\label{eq:bgb1}
F_I(t) = \max_{p\in [h_b^{-1}(t),{1\over2}]} f_I(t,p)\,,
\end{equation}
where 
\begin{align*}
f_I(x,p) &\defined \max \left\{I(W;Y)\colon I(W;X)\leq x,X\sim\mathrm{Ber}(p)\right\}\\
	&=h_b(p*\delta)-h_b\left(\delta*h_b^{-1}(h_b(p)-x)\right),
\end{align*}
that is $f_I(t,p)$ is an $F_I$-curve for a fixed marginal $P_X$.

It is sufficient to prove that $p={1\over2}$ is a maximizer in~\eqref{eq:bgb1} regardless of $t$. To that end,
recall Mrs.\ Gerber's  Lemma \cite{WZ73} states that 
$$ x \mapsto h_b(\delta * h_b^{-1}(x)) $$ 
is convex on $[0,\log 2]$. Consequently for any $0 \leq t \leq u \leq \log 2$,
$f_I(t,h_b^{-1}(u)) = 
h_b(\delta * h_b^{-1}(u)) - h_b(\delta * h_b^{-1}(u-t)) \leq h_b(\delta * h_b^{-1}(\log 2)) - h_b(\delta * h_b^{-1}(\log 2 - t)) = f_I(t,1/2)$.
\end{proof}

\section{Diagonal bound for Gaussian channels}
\label{sec:diagonal}
We now study properties of the $F_I$-curve in the Gaussian case, i.e. $P_Z=\calN(0,1).$ In this section, we show that $F_I(t,\gamma)$ is bounded away from $t$ for all $t>0$ (Theorem \ref{thm:diagonal}) and investigate the behavior of $F_I(t,\gamma)$ for small $t$ (Corollary \ref{cor:diagonalRate}). The proofs of the  non-linear SDPIs presented in both the current and the next section hinge on the existence of a linear SDPI when the input $X$ is amplitude-constrained. We define
\begin{equation}
    \label{eq:EtaDefn}
    \eta(A)\defined \sup_{P,Q~\mathrm{on}~[-A,A]}
    \frac{D(P*P_Z\|Q*P_Z)}{D(P\|Q)}.
\end{equation}
Similarly, define the Dobrushin's coefficient $\etaTV(A)$ with $D$ replaced by $\TV$ in \prettyref{eq:EtaDefn}, that is,
\begin{equation}
\etaTV(A) = 
\sup_{z,z' \in [-A,A]} \TV(P_{Z+z},P_{Z+z'}) = \sup_{|\delta|\leq 2A} \theta(\delta),
	\label{eq:etaTVA}
\end{equation}
where 
\begin{equation}
	\theta(\delta) \triangleq \TV(P_{Z},P_{Z+\delta}).
	\label{eq:theta}
\end{equation}

Observe that for any $W\to X\to Y$, where $Y=X+Z$ and $X\in [-A,A]$ almost surely, we have $I(W;Y)\leq \eta(A)I(W;X)$. In the Gaussian case considered in this section, $\eta(A)$ can be upper-bounded as  \cite{polyanskiy_dissipation_2014}
\begin{equation}
  \label{eq:etatv}
  \eta(A)\leq \etaTV(A)=\theta(A)=1-2Q(A),
\end{equation}
where $Q(x) \triangleq \int_x^\infty {1\over \sqrt{2\pi}} e^{-t^2/2} dt$ is the Gaussian complimentary CDF.
This leads to the following general lemma, which also holds for general $P_Z$.

\begin{lem}
  \label{lem:DiagonalLem}
  Let $W\to X\to Y$, where $Y=X+Z$. For any $A>0$, let $\epsilon \defined \Pr\left[|X|>A\right]$. Then
  \begin{equation}
    I(W;Y)\leq I(W;X)-\etabb(A)\left(I(W;X)-h_b(\epsilon)-\epsilon I(W;Y|E=1) \right),
  \end{equation}
  where $E\defined \ones_{\{|X|\geq A\}}$, $h_b(x)\defined x\log \frac{1}{x}+(1-x)\log \frac{1}{1-x}$ and $\etabb(A)\defined 1-\eta(A)$.
\end{lem}

\begin{proof}
Let  $\bar{\epsilon}\defined 1-\epsilon$. Then
\begin{align}
    I(W;Y)\leq& I(W;Y,E) \nonumber \\
           =&I(W;E)+\epsilon I(W;Y|E=1) 
          +\bar{\epsilon}I(W;Y|E=0) \nonumber \\
          \leq&  I(W;E)+\epsilon I(W;Y|E=1) +\bar{\epsilon}\eta(A)I(W;X|E=0), \label{eq:dissip4}
\end{align}
where the last inequality follows from the definition of $\eta(t)$ in \eqref{eq:EtaDefn}. Observing that
\begin{align*}
  \bar{\epsilon}I(W;X|E=0)=&I(W;X)-\epsilon I(W;X|E=1) -I(W;E),
\end{align*}
and denoting $\etabb(A)\defined 1-\eta(A)$, we can further bound \eqref{eq:dissip4} by
\begin{align}
  I(W;Y)\leq
          & \etabb(A)(I(W;E)+\epsilon I(W;Y|E=1) ) +\eta(A)I(W;X) +\epsilon\eta(A)\left(
          I(W;Y|E=1)-I(W;X|E=1) \right) \nonumber \\
           \leq& ~\etabb(A)\left( I(W;E)+\epsilon I(W;Y|E=1) \right) + \eta(A)I(W;X) 
           \label{eq:dissip3}\\
           =& ~ I(W;X)-\etabb(A)\left(I(W;X)- I(W;E)-\epsilon I(W;Y|E=1) \right),\nonumber
\end{align}
where \eqref{eq:dissip3} follows from $I(W;Y|E=1)\leq I(W;X|E=1)$. The result follows by noting that $I(W;E)\leq h_b(\epsilon)$.
\end{proof}

We now present explicit bounds for the value of $g_d(t,\gamma)$ when $\Expect[|X|^2]\leq 1$ and $P_Z=\calN(0,1)$.

\begin{thm} For the AWGN channel with quadratic constraint, see~\eqref{eq:figamma}, we have $F_I(t,\gamma)= t- g_d(t,\gamma)$ and
  \label{thm:diagonal}

\begin{equation}
  g_d(t,\gamma)\geq \max_{x\in[0,1/2]}2Q\left(\sqrt{\frac{\gamma}{x}}\right)\left(t-
  h_b\left(x\right)-\frac{x}{2}\log\left(
  1+\frac{\gamma}{x}
  \right) \right).
  \label{eq:vertical}
\end{equation}
\end{thm}

\begin{proof}
  Let $A \geq \sqrt{2\gamma}$ and 
	$E=\ones_{\{|X|>A/\sqrt{\gamma}\}}$ and $\EE{E}=\epsilon$. Observe that
\begin{align}
  \EE{\gamma X^2|E=1}\leq \gamma/\epsilon \mbox{ ~and~ } \epsilon\leq \gamma/A^2.
  \label{eq:Xcond}
\end{align}
Therefore, from Lemma \ref{lem:DiagonalLem} and \eqref{eq:etatv},
\begin{align}
    I(W;Y_\gamma)\leq
            I(W;X)-2Q(A)\left(I(W;X)- h_b(\epsilon)-\epsilon I(W;Y_\gamma|E=1) \right).
           \label{eq:short}
\end{align} 
Now observe that, for $\epsilon\leq\gamma/A^2\leq 1/2$, 
\begin{align}
  h_b(\epsilon)
        \leq h_b\left( \gamma/A^2 \right). \label{eq:part1}
\end{align}
In addition,
\begin{align}
  \epsilon I(W;Y_\gamma|E=1)&\leq \epsilon I(X;Y_\gamma|E=1) \nonumber \\
             &\leq \frac{\epsilon }{2}\log\left( 1+\frac{\gamma}{\epsilon} \right)
             \label{eq:GaussApprox} \\
             &\leq \frac{\gamma}{2A^2}\log(1+A^2) \label{eq:part2}.
\end{align}
Here \eqref{eq:GaussApprox} follows from the fact that mutual information
is maximized when $X$ is Gaussian under the power constraint \eqref{eq:Xcond}, and \eqref{eq:part2} follows by noticing that $x\mapsto x\log(1+a/x)$ is monotonically increasing
for any $a>0$. 
Combining \eqref{eq:short}--\eqref{eq:part2}, 
and setting $A = \sqrt{\gamma/x}$, where $0\leq x\leq 1/2$, we have
\begin{align}
  h_b(\epsilon) + \epsilon I(W;Y_\gamma|E=1) \leq h_b\left(x\right)+\frac{x}{2}\log\left(
  1+\frac{\gamma}{x}
  \right) \label{eq:workhorse2}.
\end{align}
Substituting \eqref{eq:workhorse2} in \eqref{eq:short} yields the desired result.
\end{proof}
\begin{remark}
Note that $f_d(x,\gamma)\defined
  h_b\left(x\right)+\frac{x}{2}\log\left(1+\frac{\gamma}{x}\right)$ is 0 at
  $x=0$; furthermore, $f_d(\cdot,\gamma)$ is
  continuous and strictly positive on $(0,1/2)$. Therefore $g_d(t,\gamma)$ is
  strictly positive for $t>0$. The next corollary
  characterizes the behavior of $g_d(t,\gamma)$ for small $t$.
\end{remark}

\begin{cor}
  \label{cor:diagonalRate}
    For fixed $\gamma$, $t=1/u$ and $u$ sufficiently large, there is a constant
    $c_3(\gamma)>0$ dependent on $\gamma$ such
    that
    \begin{equation}
      g_d(1/u,\gamma)\geq\frac{c_3(\gamma)}{u\sqrt{ u\gamma \log u} }e^{-\gamma u
      \log u}.
    \end{equation}
    In particular, $g_d(1/u,\gamma)\geq e^{-\gamma u\log u + O\left(\log\gamma u^{3/2} \right) }.$
    \end{cor}
\begin{proof}
Let $x = \frac{1}{2u \log u}$ in the expression being maximized in \eqref{eq:vertical}. Since $\left( \frac{x}{1+x^2}\right)e^{-x^2/2}<\sqrt{2\pi }Q(x)<x^{-1}e^{-x^2/2} $ for $x>0$, for sufficiently large $t$ 
\[
Q(\sqrt{2u\gamma \log u})= \frac{e^{-\gamma u \log u}}{2\sqrt{\pi u\gamma \log u} }+ O(\frac{e^{-\gamma u\log u}}{(u\gamma\log u)^{3/2}})
\] and
  \begin{equation}
    g_d\left( \frac{1}{2u\log u},\gamma \right)\geq \frac{3}{4u}+O\left( \frac{\log
    \log u}{u \log u} \right),
 \end{equation}
 the result follows.
\end{proof}

\begin{remark}
  \label{remark:diag}
  Fix $\gamma>0$ and define a binary random variable $X$ with $\Pr[X=a]=1/a^2$ and
  $\Pr[X=0]=1-1/a^2$ for $a>0$. Furthermore, let  $\Xh \in \{0,a\}$ denote the minimum distance 
  estimate of $X$ based on $Y_\gamma$. Then the probability of
  error satisfies $P_e = \Pr[X \neq \Xh] \leq
  Q(\sqrt{\gamma}a/2)$. In addition, $h_b\left( Q(\sqrt{\gamma}a/2)
  \right)=O(e^{-\gamma a^2/8} \sqrt{\gamma}a)$  and $H(X)=a^{-2}\log a (2+o(1))$ as $a\to
  \infty$. 
  Therefore, 
  \begin{equation}
    h_b\left( Q(\sqrt{\gamma}a/2)\right)\leq e^{-\frac{\gamma}{H(X)}\log
    \frac{1}{H(X)}+O(\log(\gamma/H(X))}.
  \end{equation}
   Using Fano's inequality, $I(X;Y_\gamma)$ can be bounded as  
  \begin{align*}
       I(X;Y_\gamma) &\geq I(X;\Xh)\\
                     &\geq H(X) - h_b(P_e)\\
                     &\geq H(X)-h_b\left( Q(\sqrt{\gamma}a/2) \right)\\
                     &=H(X)- e^{-\frac{\gamma}{H(X)}\log \frac{1}{H(X)}+O(\log(\gamma/H(X))}.
  \end{align*}
  Setting $W=X$, this result yields the sharp asymptotics \eqref{eq:FuIneq}.

\end{remark}

\begin{remark}
\label{remark:subgauss}
If the input is constrained to be subgaussian, the bound on $g_d(t)$ can be improved to polynomial in $t$. 
To see this, assume that $X$ is $s$-subgaussian, i.e. $\mathbb{P}\left[|X|>a\right]\leq \exp(-a^2/(2s))$ for $s>0$. Combining \eqref{eq:short}, \eqref{eq:part1} and \eqref{eq:GaussApprox}, we have
\begin{align}
g_d(t)\geq 2Q(A)\left(t-h_b(\epsilon)-\frac{\epsilon}{2}\log\left(1+\frac{\gamma}{\epsilon} \right) \right).
\end{align}
Since $\epsilon\leq \exp(-A^2/(2\gamma s))\defined y$, and assuming $y\leq 1/2$, the previous inequality yields 
\begin{align}
g_d(t)&\geq 2Q\left(\sqrt{2\gamma s \log y^{-1}}\right)\left(t-h_b\left(y\right)-\frac{y}{2}\log\left(1+\frac{\gamma}{y} \right) \right)\\
&\asymp \frac{y^{\gamma s}}{\sqrt{\gamma s \log y^{-1}}}\left(t + \frac{y}{2}\log y\right). \label{eq:polyapprox}
\end{align}
Choosing $y=t/\log(1/t)$ for $t \leq 1/4$, \eqref{eq:polyapprox} results in
\begin{equation}
g_d(t)\gtrsim \frac{1}{\sqrt{\gamma s}} t^{\gamma s+1} \pth{\log \frac{1}{t}}^{-(\gamma s+1/2)}.
\end{equation}
Consequently, $g_d(t)$ is polynomial in $t$, in contrast with the exponential behavior in \eqref{eq:FuIneq} under only second moment constraint.
\end{remark}

\section{Diagonal bound for general additive noise}
\label{sec:GeneralDiagonal}

In this section, we extend the diagonal bound derived in Theorem \ref{thm:diagonal} to arbitrary noise density and generalizing the power constraint to an $L_p$-norm constraint $\EE{|X|^p}\leq \gamma$ and $p\geq 1$. 

\begin{thm}
  \label{thm:DiagGeneral}
  Assume that $W\to X\to Y$, where $Y=X+Z$, $X$ and $Z$ are independent, $\EE{|X|^p}\leq \gamma$, and $Z$ has an absolute continuous distribution. Then
  \begin{equation}
    I(W;Y)\leq I(W;X) - g_d(I(W;X),\gamma),
    \label{eq:diaggeneral}
  \end{equation}
  where
  \begin{align}
    g_d(t,\gamma)\defined &~ \frac{1}{2} (1-\eta(A_2^*))t, \label{eq:gd-general}\\
    A_2^* \defined &~ \inf\left\{A>0\colon 18\gamma A^{-p}\log  (A^p)\leq t,~ A^p\geq \max\{e,2\gamma,\alpha^* e^3/ \gamma\}\right\}, \label{eq:Astar} \\
    \alpha^* \triangleq &~\inf\left\{\alpha>0\colon \eta\pth{\frac{1}{2\alpha}}\leq 1/3\right\} \label{eq:alpha}
  \end{align}
  and the amplitude-constrained contraction coefficient $\eta(\cdot)$ is defined in~\eqref{eq:EtaDefn}.
\end{thm}

\begin{cor}
\label{cor:strict}	
For any $p\geq 1$ and any $\gamma>0$, the following statements are equivalent:
\begin{enumerate}[(a)]
	\item Non-linear SDPI \prettyref{eq:diaggeneral} holds with $g_d(t,\gamma)>0$ whenever $t>0$.
	\item $S \cap (S+x)$ has non-zero Lebesgue measure for all $x\in\reals$, where $S \triangleq \{z: p_Z(z)>0\}$ is
	the support of  the probability density function $p_Z$ of $Z$. 
\end{enumerate}
\end{cor}

The proof of Theorem \ref{thm:DiagGeneral} relies on discretizing $X$. Consequently, we first derive a data processing result for the case where $X$ is an integer and a deterministic function of $W$, stated in the next lemma. We note that many steps in the proof of Lemma \ref{lem:Discrete} will be reused for deriving Theorem \ref{thm:DiagGeneral}.
\begin{lem}
  \label{lem:Discrete}
  Let $W\to X \to Y$, $Y=X+Z$, and $W\to X$ be a deterministic mapping. In addition, assume that $X$ takes values on some $\Delta$-grid for $\Delta>0$ (i.e. $X/\Delta \in \integers$ almost surely) and $\EE{|X|^p}\leq \gamma$, $p\geq 1$. Then
  \begin{equation}
     \label{eq:discretebound}
      I(W;Y)\leq I(X;Y)\leq \left(1-\frac{\etabb(A_1^*)}{2} \right)H(X),
  \end{equation}
  where
  \begin{equation}
    A_1^* \defined \min\left\{A \colon A^p\geq \max\{e,2\gamma,e^3/\gamma\Delta\}, A^{-p}\log A^p \leq \frac{H(X)}{6\gamma} \right\}
    \label{eq:a1star}
  \end{equation}

\end{lem}
\begin{proof}
Let $E\defined \ones_{\{|X|\geq A\}}$ and $\epsilon \defined \Pr\left[E=1 \right]$. Then, from Lemma \ref{lem:DiagonalLem},
\begin{align}
    I(X;Y)\leq H(X) -  \bar{\eta}(A)\left( H(X) - h_b(\epsilon) -\epsilon
            H(X|E=1)\right) \label{eq:discrete1}.         
\end{align}
Observe that for $\EE{|X|^p}\leq \gamma$,
\begin{equation}
  \epsilon = \Pr[|X|\geq A ]\leq \gamma/A^p,
\end{equation}
and, for $A\geq 1$
\begin{equation}
    \EE{|X||E=1}\leq \EE{|X|^p|E=1}\leq \gamma/\epsilon.
\end{equation}
In addition, for any integer-valued random variable $U$
we have (cf.~\cite[Lemma 13.5.4]{cover-thomas})
\begin{equation}
  H(U)\leq \left( \EE{|U|}+1 \right)h_b\left( \frac{1}{\EE{|U|}+1} \right)+\log 2.
  \label{eq:U}
\end{equation}
Consequently, for $A^p\geq  2\gamma$,
\begin{align}
    &~h_b(\epsilon)+ \epsilon H(X|E=1) \nonumber\\
    \leq &~ h_b(\epsilon)+\left(\frac{\gamma}{\Delta}+\epsilon 
    \right)h_b\left( \frac{\epsilon}{\frac{\gamma}{\Delta}+\epsilon}\right) +\epsilon\log 2 \label{eq:newstep1}\\
    \leq &~ h_b\left( \frac{\gamma}{A^p} \right)+\frac{\gamma}{A^p}\left(\frac{A^p}{\Delta} +1\right)h_b\left(
    \frac{1}{1+A^p/\Delta} \right)+\frac{\gamma}{A^p}\log 2 \label{eq:newstep2}\\
    \leq &~ \frac{\gamma}{A^p}\log A^p +\frac{\gamma}{A^p}\left(1+\log \frac{2}{\gamma}\right)+\frac{\gamma}{A^p}\left(\log\left(\frac{A^p}{\Delta}+1 \right)+\frac{A^p}{\Delta}\log\left(1+\frac{\Delta}{A^p}\right)\right) \label{eq:decomp1}\\
    \leq &~ \frac{\gamma}{A^p}\log A^p +\frac{\gamma}{A^p}\left(2+\frac{2}{\gamma}\right)+\frac{\gamma}{A^p}\log\left(\frac{A^p}{\Delta}+1 \right) \label{eq:decomp2}\\
    \leq &~ \frac{2\gamma}{A^p}\log A^p+\frac{\gamma}{A^p}\left(3+\log \frac{2}{\gamma\Delta}\right), \label{eq:decomp3}
\end{align}
where \eqref{eq:newstep1} follows from \eqref{eq:U}, $\eqref{eq:newstep2}$ follows from \eqref{eq:newstep1} being increasing in $\epsilon$ for $\epsilon\in [0,1/2]$, which is satisfied due to the assumption $A^p\geq 2\gamma$, \eqref{eq:decomp1} and  \eqref{eq:decomp2}  follows from the fact that $-(1-x)\log(1-x)\leq x$ and  $\log(x+1)\leq x$ for $x\in [0,1]$, respectively, and \eqref{eq:decomp3} follows by observing that $\log(x+1)\leq \log x +1$. Assuming $A^p\geq {e^3}/\gamma \Delta$, the last inequality yields
\begin{align}
    h_b(\epsilon)+ \epsilon H(X|E=1)\leq \frac{3\gamma \log A^p}{A^p}. \label{eq:hApprox}
\end{align}
%
Since the right-hand side of the previous equation is strictly decreasing for $A^p\geq \exp(1)$, $A$ can be chosen sufficiently large such that $\frac{3\gamma \log A^p}{A^p}\leq H(X)/2$. Choosing $A=A_1^*$, where $A_1^*$ is given in \eqref{eq:a1star}, and combining \eqref{eq:hApprox} and \eqref{eq:discrete1}, we conclude that
\begin{equation*}
   I(X;Y)\leq \left(1-\frac{\bar{\eta}(A_1^*)}{2}\right)H(X),
\end{equation*}
proving the lemma.
\end{proof}

\begin{proof}[Proof of \prettyref{thm:DiagGeneral}]
We start by verifying that $\alpha$  defined in \prettyref{eq:alpha} is finite and so is $A_2^*$ in \prettyref{eq:Astar}. Since
$\eta(a) \leq \etaTV(a)$, it suffices to show that $\etaTV(a)$ vanishes as $a\to 0$. Recall $\theta(\delta) = \frac{1}{2} \int |p_{Z}(z)- p_{Z}(z+\delta)|\diff z$ as defined in \prettyref{eq:theta}. By the denseness of compactly supported continuous functions in $L^1$, $\theta(a)\to0$ as $a\to0$. Furthermore, the translation invariance and the triangle inequality of total variation imply that $|\theta(a)-\theta(a')| \leq \theta(|a-a'|)$ and hence $\theta$ is uniformly continuous. Therefore, 
\begin{equation}
	\etaTV(a) = \max_{|\delta|\leq 2a} \theta(\delta)
	\label{eq:etaTV-max}
\end{equation}
is continuous in $a$ on $\reals_+$, which ensures that $\alpha^*$ is finite.

From Lemma \ref{lem:DiagonalLem}, and once more denoting $E\defined \ones_{\{|X|\geq A\}}$, $\epsilon \defined \Pr[|X|\geq A]$ and $\etabb(A)=1-\eta(A)$, we have
\begin{align}
  I(W;Y) 
         &   \leq I(W;X) -\bar{\eta}(A)\left( I(W;X)-h_b(\epsilon)-\epsilon
         I(W;Y|E=1) \right). \label{eq:quant1}
\end{align}
Let $Q_\alpha=\lfloor \alpha X \rfloor$.
Then
\begin{align*}
    I(W;Y) &\leq I(Q_\alpha;Y) + I(W;Y|Q_\alpha)\\
    &\leq I(Q_\alpha;Y) + \eta\left( \frac{1}{2\alpha} \right) I(W;X|Q_\alpha) \\
    &\leq H(Q_\alpha) +\eta\left( \frac{1}{2\alpha} \right)I(W;X).
\end{align*}
Thus,
\begin{equation}
  \label{eq:quant2}
    I(W;Y|E=1)\leq H(Q_\alpha|E=1) +\eta\left(  \frac{1}{2\alpha}
    \right)I(W;X|E=1).
\end{equation}
Since
\begin{equation}
  \label{eq:quant3}
    \epsilon I(W;X|E=1) \leq I(W;X),
\end{equation}
combining \eqref{eq:quant1}--\eqref{eq:quant3} gives
\begin{align}
  \label{eq:quant4}
  I(W;Y)\leq I(W;X) -\bar{\eta}(A)\left( I(W;X)-h_b(\epsilon)-\epsilon
         H(Q_\alpha|E=1) -\eta\left(  \frac{1}{2\alpha}
    \right)I(W;X) \right).
\end{align}

Since $\EE{|Q_\alpha||E=1     }\leq \alpha\gamma/A^p$, from \eqref{eq:U} and  \eqref{eq:hApprox} it follows that for $A^p\geq \alpha e^3/\gamma$,
\begin{equation}
  \label{eq:quant5}
     h_b(\epsilon)+ \epsilon H(Q_\alpha|E=1)  \leq
     \frac{3\gamma \log (A^p)}{A^p}.
\end{equation}
Thus, choosing $\alpha$ such that $\eta(1/2\alpha)\leq 1/3$, and $A$ sufficiently large such that $3\gamma A^{-p}\log A^p\leq I(W;X)/6$, \eqref{eq:quant5} becomes
 \begin{equation}
   I(W;Y)\leq I(W;X)\left( 1-\frac{\bar{\eta}(A)}{2} \right),
 \end{equation}
proving the result upon choosing $A = A_2^*$.
\end{proof}

\begin{proof}[Proof of \prettyref{cor:strict}]
To show (a) $\Rightarrow$ (b), suppose that $S \cap (S+x_0)$ has zero Lebesgue measure for some $x_0$. Consider $W=X=x_0 B$, where $B\sim\text{Bernoulli}(\epsilon)$ with $\Expect[|X|^p]=\epsilon |x_0|^p \leq \gamma$. Since $\TV(P_Z,P_{Z+z})=0$, $X$ can be perfectly decoded from $Y=X+Z$ and hence $I(W;Y)=I(W;X)=H(X)$, which shows that $F_I(t) = t$ in a neighborhood of zero.

To show (b) $\Rightarrow$ (a), 
in view of \prettyref{thm:DiagGeneral}, it suffices to show that $\eta(A)<1$ for all finite $A$.
Recall that for any channel, $ \etaKL = 1$ if and only if $\etaTV = 1 $ (\cite[Proposition II.4.12]{CKZ98}). 
Therefore it is equivalent to show that  $\etaTV(A)<1$ for all finite $A$.
Suppose otherwise, \ie, $\etaTV(A)=1$ for some $A>0$. By \prettyref{eq:etaTV-max}, there exists some $\delta\in[-A,A]$ such that $\TV(P_Z,P_{Z+\delta})=1$, which means that $S \cap (S+\delta)$ has zero Lebesgue, contradicting the assumption (b) and completing the proof.
\end{proof}


\section{Minimum mean square error and near-Gaussianness}
\label{sec:mmse}
We now take a step back from strong data-processing inequalities and present an ancillary result of independent interest.  We  prove that any random variable for which there exists an almost optimal (in terms of the mean-squared error) linear estimator operating on the Gaussian-corrupted measurement must necessarily be almost Gaussian (in terms of the Kolmogorov-Smirnov distance). We will use this result in the next section to bound the horizontal gap $g_h(t,\gamma)$ for Gaussian noise.
%
%

Throughout the rest of the paper we make use of Fourier-analytic tools and, in particular,  Esseen's inequality, stated below for reference.

\begin{lem}[{\cite[Eq.~(3.13), p.~538]{feller_introduction_1966}}]
\label{lem:Esseen}
Let $P$ and $Q$ be two distributions with characteristic functions $\cf_P$ and $\cf_Q$, respectively. In addition, assume that $Q$ has a bounded density $q$. Then 
\begin{equation}
  \dks(P,Q)\leq \frac{1}{\pi}\int_{-T}^T \left| \frac{\cf_P(\omega)-\cf_Q(\omega)}{\omega} \right|\diff \omega +\frac{24 \|q\|_\infty}{\pi T},
\end{equation}
where 
$ \dks(P,Q)\defined \sup_{x \in \Reals} |F_P(x)-F_Q(x)|$ is the Kolmogorov-Smirnov distance.
\end{lem}

We show next that if the linear least-square error of estimating $X$ from $Y_\gamma$ is small (i.e. close to the minimum mean-squared error), then $X$ must be almost Gaussian in terms of the KS-distance. With this result in hand, we use the I-MMSE relationship  \cite{guo_mutual_2005}  to show that if $I(X;Y_\gamma)$ is close to $C(\gamma)$, then $X$ is also almost Gaussian. This result, in turn, will be applied in the next section to bound $F_I(t,\gamma)$ aways from $C(\gamma)$.

Let $P_Z = \calN(0,1)$,  $\EE{|X|^2}= 1$ and $\EE{X}=0$. Denote the linear least-square error estimator of
$X$ given $Y_\gamma$ by 
   $f_{L}(y) \defined\sqrt{ \gamma}y/(1+\gamma)$,
whose mean-squared error is
\begin{equation*}
   \lse(X|Y_\gamma)\defined \EE{(X-f_{L}(Y_\gamma))^2}=\frac{1}{1+\gamma}.
\end{equation*}

Assume that $\lse(X|Y_\gamma)-\mmse(X|Y_\gamma) \leq \epsilon$. It is well known that $\epsilon=0$ if and only if $X\sim \calN(0,1)$ (see e.g. \cite{guo_estimation_2011}). To develop a finitary version of this result, we ask the following question: If $\epsilon$ is small, how close is $P_X$ to Gaussian? The next lemma provides a quantitative answer.
\begin{lem}
\label{lem:KSMMSE}
    For $\EE{|X|^2}= 1$ and $\EE{X}=0$, if $\lse(X|Y_\gamma)-\mmse(X|Y_\gamma)\leq \epsilon$, then 
    there are absolute constants $a_0$ and $a_1$ such that
    \begin{align}
        \dks(F_X,\mathcal{N}(0,1))\leq& a_0\sqrt{\frac{1}{\gamma
        \log(1/\epsilon)}} +a_1(1+\gamma)\epsilon^{1/4}\sqrt{\gamma\log(1/\epsilon)}.
    \end{align}
\end{lem}

%
\begin{proof}
Denote $f_M(y)=\EE{X|Y_\gamma=y}$. Then
\begin{align*}
    \epsilon &\geq \lse(X|Y_\gamma)-\mmse(X|Y_\gamma)\\
    &=\EE{(X-f_{L}(Y_\gamma))^2}-\EE{(X-f_{M}(Y_\gamma))^2}\nonumber \\
                  &= \EE{(f_M(Y_\gamma)-f_L(Y_\gamma))^2}.
\end{align*}
Denote $\Delta(y)\defined f_M(y)-f_L(y)$. Then $\EE{\Delta(Y_\gamma)}=0$ and $\EE{\Delta(Y_\gamma)^2}\leq \epsilon$. From the orthogonality principle:
\begin{align}
    \EE{e^{itY_\gamma}(X-f_M(Y_\gamma))}=0.
\end{align}
Let $\varphi_X$ denote the characteristic function of $X$. Then
\begin{align}
    \EE{e^{itY_\gamma}(X-f_M(Y_\gamma))} =&
    \EE{e^{itY_\gamma}(X-f_L(Y_\gamma)-\Delta(Y_\gamma))} \nonumber \\
 =&\frac{1}{1+\gamma}\left(e^{-t^2/2}\EE{e^{i\sqrt{\gamma}tX}X}-\sqrt{\gamma}\varphi_X(\sqrt{\gamma
 }t)\EE{Ze^{itZ} }\right)-\EE{e^{itY_\gamma}\Delta(Y_\gamma)} \nonumber\\
 =&\frac{-ie^{-u^2/2\gamma}}{1+\gamma}\left(\varphi_X'(u)+u\varphi_X(u)\right)-\EE{e^{itY_\gamma}\Delta(Y_\gamma)},
\end{align}
where the last equality follows by changing variables $u=\sqrt{\gamma}t$. Consequently,
\begin{align}
    \frac{e^{-u^2/2\gamma}}{1+\gamma}\left| \varphi_X'(u)+u\varphi_X(u) \right| &= \left| \EE{e^{itY_\gamma}\Delta(Y_\gamma)} \right|\\
    &\leq \EE{\left|\Delta(Y_\gamma)\right|}\nonumber\\
    &\leq \sqrt{\epsilon}. \label{eq:boundODE}
\end{align}
Put $\varphi_X(u)=e^{-u^2/2}\left(1+z(u) \right)$. Then
\begin{align}
    \left| \varphi_X'(u)+u\varphi_X(u) \right|=e^{-u^2/2}|z'(u)|, \nonumber
\end{align}
and, from \eqref{eq:boundODE},
    $\left|z'(u)\right|\leq
    (1+\gamma)\sqrt{\epsilon}e^{\frac{u^2(\gamma+1)}{2\gamma}}.$ 
Since $z(0)=0$,
\begin{align}
    |z(u)|&\leq \int_0^u |z'(x)|dx
          \leq u(1+\gamma)\sqrt{\epsilon}e^{\frac{u^2(\gamma+1)}{2\gamma}} \label{eq:hubound}.
\end{align}
Observe that $|\varphi_X(u)-e^{-u^2/2 }|=e^{-u^2/2}|z(u)|$. Then, from \eqref{eq:hubound},
\begin{equation}
    \left|\frac{ \varphi_X(u)-e^{-u^2/2}}{ u} \right|\leq (1+\gamma)\sqrt{\epsilon}e^{\frac{u^2}{2\gamma}}.
\end{equation}
Thus, Lemma \ref{lem:Esseen} yields
\begin{align*}
    \dks(F_X,\calN(0,1))&\leq  \frac{1}{\pi}\int^T_{-T} (1+\gamma)\sqrt{\epsilon}e^{\frac{u^2}{2\gamma}} du +\frac{12\sqrt{2}}{\pi^{3/2}T}\\
                         &\leq \frac{2T}{\pi}   (1+\gamma)\sqrt{\epsilon}e^{\frac{T^2}{2\gamma}}+\frac{12\sqrt{2}}{\pi^{3/2}T}.
\end{align*}
Choosing $T=\sqrt{\frac{\gamma}{2}\log(\frac{1}{\epsilon})} $, we find
\begin{align*}
    \dks(F_X,\mathcal{N}(0,1))\leq& a_0\sqrt{\frac{1}{\gamma
    \log(1/\epsilon)}}+a_1(1+\gamma)\epsilon^{1/4}\sqrt{\gamma\log(1/\epsilon)},
\end{align*}
where $a_0=\frac{24}{\pi^{3/2}} $ and $a_1=\frac{\sqrt{2}}{\pi}$.
\end{proof}

Through the I-MMSE relationship \cite{guo_mutual_2005}, the previous lemma can be extended
to bound the KS-distance between the distribution of $X$ and the
Gaussian distribution when $I(X;Y_\gamma)$ is close to $C(\gamma)$.

\begin{lem}
  \label{lem:capKS}
Assume that $\EE{|X|^2}= 1$, $\EE{X}=0$, and $C(\gamma)-I(X;Y_\gamma)\leq \epsilon$. Then, for $\gamma>4\epsilon$,
\begin{align}
    \dks(F_X,\mathcal{N}(0,1))\leq& a_0\sqrt{\frac{2}{\gamma
      \log\left(\frac{\gamma}{4\epsilon}\right)}} +a_1(1+\gamma)(\gamma\epsilon)^{1/4}\sqrt{2\log\left(\frac{\gamma}{4\epsilon}\right)}.
\end{align}
\end{lem}

\begin{proof}
From the I-MMSE relationship \cite{guo_mutual_2005}:
\begin{align}
    C(P)-I(X;Y_P) = \frac{1}{2}\int_0^P \frac{1}{1+\gamma}-\mmse(X|Y_\gamma)
    d\gamma \leq \epsilon.
\end{align}
Since $\mmse(X|Y_\gamma)\leq \frac{1}{1+\gamma}$, for any $\delta \in [0,P)$
\begin{align}
   \frac{1}{\delta} \int_{P-\delta}^P \frac{1}{1+\gamma}-\mmse(X|Y_\gamma)
   d\gamma \leq \frac{2\epsilon}{\delta}.
\end{align}
The function $\mmse(X|Y_\gamma)$ is continuous in $\gamma$. Then, from the
mean-value theorem for integrals, there exists $\gamma^* \in (P-\delta,P)$ such
that
\begin{equation}
  \frac{1}{1+\gamma^*}-\mmse(X|Y_{\gamma^*})\leq \frac{2\epsilon}{\delta}. 
\end{equation}
From Lemma \ref{lem:KSMMSE}, we find
\begin{align*}
     d_{KS}(F_X,\mathcal{N}(0,1))\leq& a_0\sqrt{\frac{1}{\gamma^*
     \log(\delta/2\epsilon)}}+a_1(1+\gamma^*)\left(\frac{2\epsilon}{\delta}\right)^{1/4}\sqrt{\gamma^*\log(\delta/2\epsilon)}\\
     \leq& a_0\sqrt{\frac{1}{(P-\delta) \log(\delta/2\epsilon)}}+a_1(1+P)\left(\frac{2\epsilon}{\delta}\right)^{1/4}\sqrt{P\log(\delta/2\epsilon)}.
\end{align*}
The desired result is found by choosing $\delta = P/2$.
%
\end{proof}


\begin{remark}
Note that the gap between the linear and nonlinear MMSE can be expressed as the Fisher distance between the convolutions, i.e.,
$\lse(X|Y_\gamma) - \mmse(X|Y_\gamma) = I(P_X * N(0,1) \| N(0,1+\gamma))$, where $I(P\|Q) = \int [(\log \fracd{P}{Q})']^2 \diff P$ is the Fisher distance. Similarly, $C(\gamma)-I(X;Y_\gamma) = D(P_X*N(0,1)\|N(0,1+\gamma))$. Therefore Lemma~\ref{lem:KSMMSE} (resp., Lemma~\ref{lem:capKS}) can be interpreted as a deconvolution result, where bounds on a stronger Fisher (resp.~KL) distance between the convolutions lead to bounds on the distance between the original distributions under a weaker (KS) metric. Recall also that Gross's log-Sobolev inequality bounds KL in terms of Fisher distance.
  \label{remark:deconv}
\end{remark}



\section{Horizontal bound for Gaussian channels}
\label{sec:horizontal}

Using the results from the previous section, we show that, for $P_Z\sim \calN(0,1)$, $F_I(t,\gamma)$ is bounded away from the capacity $C(\gamma)$ for all $t$.
 \begin{thm}
   \label{thm:capBound}
For the AWGN channel with quadratic constraint, see~\eqref{eq:figamma}, we have $F_I(t,\gamma)=  C(\gamma) -
g_h(t,\gamma)$ and 
	$$ g_h(t,\gamma)\geq e^{-c_1(\gamma)e^{4t}}\,,$$
where $c_1(\gamma)$ is some positive constant depending on $\gamma$.
 \end{thm}

We first give an auxiliary lemma.
\begin{lem}
  \label{lem:concentration}
    If $D(\calN(0,1)\|P_X*\calN(0,1))\leq 2 \epsilon$ for $\epsilon\leq 1$, then there exists an
    absolute constant $a_2>0$ such that
    \begin{equation}
      \Pr[ |X|> \epsilon^{1/8} ]\leq a_2\epsilon^{1/8}.
    \end{equation}
\end{lem}
\begin{proof}
 Let $Z\sim \calN(0,1)\independent X$, and $B(x,\delta)\defined [x-\delta,x+\delta]$. For any $\delta\in (0,1)$, Pinsker's
 inequality yields
   \begin{align*}
       \Pr[Z\in B(0,\delta)]-\Pr[Z+X\in B(0,\delta)]&\leq
       \TV(P_Z,P_{Z+X})\\
       &\leq \sqrt{\epsilon}.
   \end{align*}
 Observe that 
    \begin{align*}
        \Pr[Z+X\in B(0,\delta)]= &\Pr\left[Z\in B(-X,\delta)\mid |X|\leq
        3\delta\right]\Pr[|X|<3\delta]+\Pr\left[Z\in B(-X,\delta)\mid
        |X|> 3\delta\right]\Pr[|X|>3\delta]\\ 
        \leq& \Pr\left[Z\in
        B(0,\delta)\right]\Pr[|X|\leq 3\delta]+\Pr\left[Z\in
        B(3\delta,\delta)\right]\Pr[|X|>3\delta]\\
        =&\Pr\left[ |X|>3\delta
        \right]\left( \Pr\left[ Z\in B(3\delta,\delta)-\Pr[Z\in B(0,\delta)]
        \right] \right)+ \Pr\left[ Z\in B(0,\delta) \right].
    \end{align*}
 Consequently,
 \begin{align*}
    \Pr\left[ |X|>3\delta
        \right]\left( \Pr[Z\in B(0,\delta)]-\Pr\left[ Z\in B(3\delta,\delta)
        \right] \right)\leq \sqrt{\epsilon}
 \end{align*}
 Let $\phi(u)\defined\frac{1}{\sqrt{2\pi}}e^{-u^2/2}$ be the Gaussian probability density function. For $\delta\leq 1/2$
 \begin{align*}
     \Pr[Z\in B(0,\delta)]-\Pr\left[ Z\in B(3\delta,\delta)
        \right]& \geq 2\delta (\phi(\delta)-\phi(2\delta))\\
        &\geq \frac{1}{4}\delta^3,
  \end{align*}
 where the last inequality follows from the fact that the mapping $u\mapsto e^{-u^2/2}-e^{-(2u)^2/2}$ is lower-bounded by $u^2/2$ for $\delta\leq 1/2$. Then
 \begin{align}
   \Pr[|X|>3\delta]\leq 4\delta^{-3}\epsilon 
 \end{align}
 The result follows by choosing $\delta = \frac{\epsilon^{1/8}}{3}$ with
 constant $a_2=108$.
 \end{proof}

\begin{proof}[Proof of \prettyref{thm:capBound}]
We will show an equivalent statement: If $t>0$ is such that $C(\gamma)-F_I(t,\gamma)\leq \epsilon$ then 
   \begin{equation}
   \label{eq:lnln}
   t \geq \frac{1}{4}\log\log\frac{1}{\epsilon} -  \log c_1(\gamma).
 \end{equation}
We assume without loss of generality that $\EE{|X|^2}=1$,  $\EE{X}=0$, and $\epsilon<1$. If $\EE{|X|^2}= \sigma^2>0$, the following derivation holds by appropriately scaling the parameter $\gamma$, i.e. replacing $\gamma$ by $\gamma' = \sigma^2\gamma$, without changing the asymptotic scaling of the results.  
    Observe that from the saddle-point property of the AWGN
    \begin{align}
        I(W;Y_\gamma) =& I(X;Y_\gamma)-I(X;Y_\gamma|W) \nonumber\\
               \leq& C(\gamma)-D(P_{\sqrt{\gamma}X}*\calN(0,1)\|\calN(0,1+\gamma))
              -I(X;Y_\gamma|W).
               \label{eq:gap}
    \end{align}
Therefore, if $I(W;Y_\gamma)$ is close to $C(\gamma)$, then (a) $P_X$ needs to be Gaussian
like, and (b) $P_{X|W}$ needs to be almost deterministic with high
$P_W$-probability. Consequently, $P_{X|W}$ and $P_X$ are close to being mutually
singular and hence $I(W;X)$ will be large, since
\begin{equation*}
    I(W;X)=D(P_{X|W}\|P_X|P_W).
\end{equation*}

Let $\Xt\defined
\sqrt{\gamma}X$ and then $W \to \Xt \to Y_\gamma$. Define
\begin{align}
d(x,w)
\triangleq & ~ D(P_{Y_\gamma|\Xt=x}\|P_{Y_\gamma|W=w}) 	\nonumber \\
= & ~ 	D(\calN(x,1)\|P_{\Xt|W=w}*\calN(0,1)). 	\label{eq:dxw}
\end{align}
Then $(x,w) \mapsto d(x,w)$ is jointly measurable\footnote{By definition of the Markov kernel, both $x \mapsto P_{Y_{\gamma}\in A|\Xt=x}$ and $w \mapsto P_{Y_{\gamma}\in A|W=w}$ are measurable for any measurable subset $A$. 
Let $[y]_k \triangleq \floor{k y}/k$ denote the uniform quantizer.
By the data processing inequality and the lower semicontinuity of divergence, we have 
$D(P_{[Y_{\gamma}]_k|\Xt=x}\|P_{[Y_{\gamma}]_k|W=w}) \to D(P_{Y_{\gamma}|\Xt=x}\|P_{Y_{\gamma}|W=w})$ as $k\to\infty$. Therefore the joint measurability of $(x,w)\mapsto D(P_{Y_{\gamma}|\Xt=x}\|P_{Y_{\gamma}|W=w})$ follows from that of $(x,w)\mapsto D(P_{[Y_{\gamma}]_k|\Xt=x}\|P_{[Y_{\gamma}]_k|W=w})$.}
and $I(X;Y_\gamma|W)=\Expect[d(\Xt,W)]$.
Similarly, $w \mapsto \tau(w) \triangleq D(P_{X|W=w}\|P_{X})$ is measurable and $I(X;W)=\Expect[\tau(W)]$.
Since $ \epsilon \geq I(X;Y_\gamma|W) $ in view of \prettyref{eq:gap}, we have
\begin{align} 
\epsilon &\geq \mathbb{E}[d(\Xt,W)] \geq 2\epsilon\cdot\Pr[ d(\Xt,W)\geq 2 \epsilon ]. 
\end{align}
Therefore 
\begin{align}
  \Pr[ d(\Xt,W) < 2 \epsilon ] > \frac{1}{2}. \label{eq:probB}
\end{align}
In view of Lemma
\ref{lem:concentration}, if $d(x,w) < 2 \epsilon$, then
\begin{align*}
   \Pr[ \tilde{X}\in
  B(x,\epsilon^{1/8})|W=w ]= \Pr\left[ X\in
  B\left(\frac{x}{\sqrt{\gamma}},\frac{\epsilon^{1/8}}{\sqrt{\gamma}}\right)\middle| W=w \right] \geq 1-a_2\epsilon^{1/8}.
\end{align*}
Therefore, with probability at least $1/2$, $\Xt$ and, consequently, $X$ is
concentrated on a small ball.  
Furthermore, Lemma \ref{lem:capKS} implies
that there exist absolute constants $a_3$ and $a_4$ such that
\begin{align*}
  \Pr\left[ X\in
  B\left(\frac{x}{\sqrt{\gamma}},\frac{\epsilon^{1/8}}{\sqrt{\gamma}}\right)  \right]&\leq
  \Pr\left[ Z\in
  B\left(\frac{x}{\sqrt{\gamma}},\frac{\epsilon^{1/8}}{\sqrt{\gamma}}\right)
\right]+2\dks(F_X,\calN(0,1))\\
&\leq  \frac{\sqrt{2}\epsilon^{1/8}}{\sqrt{\pi\gamma}}+ a_3\sqrt{\frac{1}{\gamma
  \log\left(\frac{\gamma}{4\epsilon}\right)}}+a_4(1+\gamma)(\gamma\epsilon)^{1/4}\sqrt{\log\left(\frac{\gamma}{4\epsilon}\right)}\\
  &\leq \kappa(\gamma)\left( \log\frac{1}{\epsilon} \right)^{-1/2},
\end{align*}
where $\kappa(\gamma)$ is some positive constant depending only on $\gamma$. 
Therefore, for 
$\epsilon$ sufficiently small, denoting $E= B(\frac{x}{\sqrt{\gamma}},\frac{\epsilon^{1/8}}{\sqrt{\gamma}})$, we have by data processing inequality:
for any $w$ in the support of $W$,
\begin{align}
\tau(w)&=  D(P_{X|W=w}\|P_X)\\
&\geq P_{X|W=w}(E)\log \frac{ P_{X|W=w}(E)}{ P_{X}(E)}+  P_{X|W=w}(E^{c})\log \frac{ P_{X|W=w}(E^{c})}{ P_{X}(E^{c})}
 \nonumber \\
  &\geq \frac{1}{2}\log\log\frac{1}{\epsilon} -\log \kappa(\gamma)-a_5, \label{eq:explode}
\end{align}
where $a_5$ is an absolute positive constant. Combining
\eqref{eq:explode} with \eqref{eq:probB} and letting  $c_1^2(\gamma)\defined
e^{a_5}\kappa(\gamma)$, we obtain
\begin{equation}
\label{eq:ExplodeGauss}
\Pr\Big[\tau(W) \geq \frac{1}{2}\log\log\frac{1}{\epsilon} - 2 \log c_1(\gamma)\Big]
\geq \Pr[d(\tX,W) < 2 \epsilon] \geq \frac{1}{2},
\end{equation}
which implies that $I(W;X)=\Expect[\tau(W)] \geq  \frac{1}{4}\log\log\frac{1}{\epsilon} - \log c_1(\gamma)$,
proving the desired \prettyref{eq:lnln}.
\end{proof}

\begin{remark}
  \label{remark:cap}
The double-exponential convergence rate in Theorem~\ref{thm:capBound} is in fact sharp. To see this, note that \cite[Theorem 8]{wu_impact_2010} showed that there exists a sequence of zero-mean and unit-variance random variables $X_m$ with $m$ atoms,  such that
  \begin{equation}
    C(\gamma)-I(X_m;\sqrt{\gamma} X_m + Z)\leq 4(1+\gamma)\left( \frac{\gamma}{1+\gamma} \right)^{2m}.
  \end{equation}
  Consequently,
  \begin{align*}
    C(\gamma)-F_I(t,\gamma)&\leq C(\gamma)-F_I(\log \lfloor e^t
    \rfloor,\gamma)\\
    &\leq 4(1+\gamma)\left( \frac{\gamma}{1+\gamma} \right)^{2(e^t-1)}\\
    &= e^{-2e^t\log \frac{1+\gamma}{\gamma} +O(\log \gamma) },
  \end{align*}
  proving the right-hand side of \eqref{eq:bound_gh}.
\end{remark}


\section{Deconvolution results for total variation}
\label{sec:deconvTV}

The proof of the horizontal gap for the scalar AWGN channel in \prettyref{sec:horizontal} consists of four steps:
\begin{enumerate}[(a)]
	\item Notice that if $C(\gamma)-I(W;Y_\gamma)$ is small, then both $X$ is Gaussian-like and $P_X$ and $P_{X|W}$ are close to being mutually singular;
	\item Use Lemma \ref{lem:capKS} to show that $ P_X$ cannot be concentrated on any ball of small radius if it is Gaussian-like;
	\item Apply Lemma \ref{lem:concentration} to show that $P_{X|W}$, in turn, is  concentrated on a small ball with high $W$-probability;
	\item Use \eqref{eq:ExplodeGauss} to show that $I(W;X)$ must explode.
\end{enumerate}

In \prettyref{sec:GeneralHorizontal}, we will implement the above program to extend the results in Theorem \ref{thm:capBound} (i.e. $I(W;Y)$ approaches capacity only as $I(W;X)\to \infty$) for a  range of noise distributions.  We also generalize the moment constraint on the input distribution, allowing $P_X$ to be restricted to an  arbitrary convex set. However, the extension of the AWGN result to a wider class of noise distributions requires new deconvolution results that are similar in spirit to Lemmas \ref{lem:capKS} and   \ref{lem:concentration}. These results are the focus of the present section.

If $\calP$ is convex and $C(\calP)\defined \sup_{P_X\in \calP} I(X;Y)<\infty$, then there exists a unique capacity-achieving output distribution $P_{Y^*}$ \cite{kemperman1974shannon}. In addition, by the saddle-point characterization of capacity, $$ C(\calP) = \sup_{P_X\in \calP} D(P_{Y|X}\| P_Y^*|P_X).$$ Consequently, for any $P_X\in \calP$, we can  decompose
\begin{equation}
	\label{eq:IDecomp}
  I(W;Y)=I(X;Y)-I(X;Y|W)\leq C(\calP)-D(P_Y\| P_Y^*)-I(X;Y|W).
\end{equation}
If the capacity-achieving input distribution $P_{X^*}$ is unique, then the same intuition for the Gaussian case should hold: (i) $P_X$ must be close to the capacity achieving input distribution $P_{X^*}$ and (ii) $P_{X|W}$ must be concentrated on a small ball with high probability. Therefore, as long as $P_{X^*}$ is assumed to have no atoms, then $P_{X|W}$ and $P_X$  are close to being mutually singular, which, in view of the fact that 
\begin{equation}
\label{eq:IWX}
I(W;X)=D(P_{X|W}\| P_X|P_W),
\end{equation}
implies that $I(W;X)$ will explode.

In order to make this proof concrete, we require additional results to quantify the distance between $P_X$ and $P_{X^*}$ (analogous to Lemma \ref{lem:capKS} in the Gaussian case), and to show that $ P_{X|W} $ is concentrated in a small ball (analogous to Lemma \ref{lem:concentration}) for general $P_Z$. These are precisely the results we present in this section, once again making use of Lemma \ref{lem:Esseen} and Fourier-analytic tools. In particular, we prove a deconvolution result in terms of total variation for a wide range of additive noise distributions $P_Z$ (e.g. Gaussian, uniform). The main result in this section (\prettyref{thm:KS_TV} and \prettyref{cor:deconv-phiz}) states that, under first moment constraints and certain conditions on the characteristic function of $P_Z$ (\eg, no zeros, cf. Lemma \ref{lem:DeconvGen}), if $\TV(P*P_Z,Q*P_Z)$ is small and $Q$ has a bounded density, then $\dks(P,Q)$ is also small.

Let $v:\Reals\to \Reals$ be the positive, symmetric function
\begin{equation}
    \label{eq:vdef}
    v(x)\defined \frac{2(1-\cos x)}{x^2}
\end{equation}
and $\hat{v}$ its Fourier transform
\begin{equation}
\label{eq:vhat}
  \hat{v}(\omega) \defined \int v(x)e^{i\omega x} dx = 2\pi\left( 1-|\omega|
  \right)^+,
\end{equation}
where $(x)^+ \triangleq \max\{x,0\}$. One of the motivations behind introducing the function $v$ in \eqref{eq:vdef} is that it enables the tail of any real-valued random variable $X$ to be bounded by its characteristic function (cf. the proof of \cite[Lemma 4.1]{kallenberg_foundations_1997}). 

We have the following deconvolution lemma.
\begin{lem}
\label{lem:DeconvGen}
    Assume $P_Z$ has density bounded by $m_1$ and that there exists a decreasing
  function $g_1:(0,1]\to \Reals^+$ with $g_1(0+)=\infty$ such that
  \begin{equation}
  \label{eq:condZ}
    \mathrm{Leb}\left\{ \omega: |\cf_Z(\omega)|\leq \sqrt{u},|\omega|\leq
  g_1(u) \right\}\leq \sqrt{g_1(u)} \qquad \forall u\in (0,1].
  \end{equation}
  Then for all distributions $P,Q$ and all $x_0\in \Reals$:
  \begin{equation}
    \label{eq:vPQ}
    \left|\ExpVal{P}{v(TX-x_0)}-\ExpVal{Q}{v(TX-x_0)} \right|\le {c\over
      \sqrt{T}},\qquad T=g_1\left( m_1\TV(P*P_Z,Q*P_Z) \right),
  \end{equation}
  where $c$ is an absolute constant.
\end{lem}

\begin{remark}
\label{rmk:g1}
\begin{enumerate}[1.]
\item The implication of the previous lemma is that $P$ and $Q$ are almost the same on all balls of size approximately $\frac{1}{T}$.

\item For Gaussian $P_Z$, $g_1(u)=\sqrt{-\log u}$. For uniform $P_Z$, $g_1(u)=u^{-1/3}$.

\item Without assumptions similar to those of \prettyref{lem:DeconvGen}, it is impossible to have any deconvolution inequality. For example, if $\cf_Z=0$ outside of a neighborhood of 0 (e.g. $p_Z$ is proportional to \eqref{eq:vdef}), then one may have $P*P_Z=Q*P_Z$, but $P\neq Q$.
\end{enumerate}
\end{remark}

\begin{proof}
Denote the density of $Z$ by $p_Z$.
   From Plancherel's theorem, we have
   \begin{align}
     \|(\cf_P-\cf_Q )\cf_Z\|_2^2 =  &  ~ 2 \pi \|P*P_Z - Q*P_Z\|_2^2 \nonumber \\
\leq &~   2 \pi \|P*P_Z - Q*P_Z\|_1  \|P*P_Z - Q*P_Z\|_\infty \nonumber \\
\leq &~   4 \pi m_1 \TV(P*P_Z , Q*P_Z) \triangleq 4\pi \delta, 
        \label{eq:TV1}
   \end{align}    
where the first inequality follows from H\"{o}lder's inequality, and the second inequality follows from 
  $\|P*P_Z - Q*P_Z\|_\infty \leq \max\{\|P*P_Z\|_\infty, \|Q*P_Z\|_\infty\} \leq \|p_Z\|_\infty$.

Assume there exist positive functions $g$ and $h$ and $T>0$ such that
\begin{equation}
  \label{eq:ghIneq}
  \mathrm{Leb}\left(\left\{ \omega: |\cf_Z(\omega)|\leq g(T),|\omega|\leq T
  \right\}\right)\leq h(T).
\end{equation}
Put $\calD \defined \left\{ \omega: |\cf_Z(\omega)|\leq g(T),|\omega|\leq T
  \right\}$ and $\calD^c = [-T,T]\backslash \calD$. Then
\begin{align}
    \frac{1}{T}\int_{-T}^T |\cf_P(\omega)-\cf_Q(\omega)|d\omega = & ~ \frac{1}{T} \left(\int_\calD |\cf_P(\omega)-\cf_Q(\omega)|d\omega +\int_{\calD^c}|\cf_P(\omega)-\cf_Q(\omega)|d\omega\right) \nonumber\\
    \overset{\mathclap{\prettyref{eq:ghIneq}}}{\leq} &~\frac{2h(T)}{T}+\frac{1}{T}\int_{\calD^c}|\cf_P(\omega)-\cf_Q(\omega)|\left(\frac{|\cf_Z(\omega) |}{g(T)}\right)d\omega \nonumber\\
    \leq &~\frac{2h(T)}{T}+\frac{1}{Tg(T)} \int_{-T}^T|\cf_P(\omega)-\cf_Q(\omega)| |\cf_Z(\omega) | d\omega \nonumber\\
    \leq &~\frac{2h(T)}{T}+\frac{\sqrt{2}\| (\cf_P-\cf_Q )\cf_Z\|_2}{g(T)\sqrt{T}} \nonumber\\
    \overset{\mathclap{\prettyref{eq:TV1}}}{\leq} &~\frac{2h(T)}{T}+\frac{\sqrt{8\pi \delta}}{\sqrt{T}g(T)}, \label{eq:TV2}
\end{align}
where 
the third inequality follows Cauchy-Schwarz inequality.

Note that it is sufficient to consider $x_0=0$, since otherwise we can simply shift the distributions $P$ and $Q$ without affecting the value of $\delta$. In addition, Plancherel's theorem and \prettyref{eq:vhat} yield
\begin{equation}
\ExpVal{P}{v(TX)}=\frac{1}{T}\int_{-T}^T \cf_P(\omega)\left(1-\frac{|\omega|}{T} \right)d\omega.	
	\label{eq:Ev}
\end{equation}
Thus, we have
\begin{align*}
 \left|\ExpVal{P}{v(TX)}-\ExpVal{Q}{v(TX)} \right|&\leq \frac{1}{T}\int_{-T}^T |\cf_P(\omega)-\cf_Q(\omega) |d\omega\\
    &\leq \frac{2h(T)}{T}+\frac{\sqrt{8 \pi \delta}}{\sqrt{T}g(T)}. 
\end{align*}
Finally, choosing $ T=g_1(\delta)$, $h(T)=\sqrt{T}$ and $g(T)=\sqrt{\delta}$, the result follows.
\end{proof}

The methods used in the proof of the previous theorem and, in particular, Eq.~\eqref{eq:TV2}, can be used to bound the KS-distance between $P$ and $Q$, as demonstrated in the next theorem.

\begin{thm}
\label{thm:KS_TV}
Assume $P_Z$ has density bounded by $m_1$ and that there exists functions $g(T)$ and $h(T)$ that satisfy 
\begin{equation}
  \mathrm{Leb}\left(\left\{ \omega: |\cf_Z(\omega)|\leq g(T),|\omega|\leq T
  \right\}\right)\leq h(T). \tag{\ref{eq:ghIneq}}
\end{equation}
Then for any pair of distributions $P$, $Q$ where $Q$ has a density bounded by $m_2$ we get for all $T>0$:
\begin{equation}
  \label{eq:detailTV_KS}
  \dks(P,Q) \leq\frac{ Th(T)}{\pi} + \frac{24 m_2 + 2 (\ExpVal{P}{|X|} + \ExpVal{Q}{|X|})}{ \pi T} +\frac{(2T)^{3/2}}{\sqrt{\pi}  g(T)}\sqrt{m_1\TV(P*P_Z,Q*Q_Z)},
\end{equation}
\end{thm}

\begin{proof}
\begin{align}
  \int_{-T}^T \frac{\left| \cf_P(\omega) -\cf_Q(\omega)\right|}{|\omega|}d\omega &\leq T \int_{|\omega|\geq1/T} \left|
  \cf_P(\omega) -\cf_Q(\omega)\right|d\omega + \int_{-1/T}^{1/T} \frac{\left| \cf_P(\omega)
  -\cf_Q(\omega)\right|}{|\omega|}d\omega \label{eq:yp1}\\
  &\leq T \int_{-T}^T \left| \cf_P(\omega) -\cf_Q(\omega)\right|d\omega + \frac{2\left( \ExpVal{P}{|X|} + \ExpVal{Q}{|X|} \right)}{T}\\
  &\leq T h(T)+\frac{T^{3/2}\sqrt{8\pi\delta}}{g(T)}+ \frac{2\left( \ExpVal{P}{|X|} + \ExpVal{Q}{|X|} \right)}{T},
\end{align}
where the first inequality follows from $\frac{1}{|\omega|}\leq T$ for $|\omega|\geq T$, the second inequality follows from the triangle inequality and the fact that $|\cf_P(\omega)-1|\leq
|\omega|\ExpVal{P}{|X|}$, and the last inequality follows from \eqref{eq:TV2}. Using Lemma \ref{lem:Esseen}, we
get~\eqref{eq:detailTV_KS}.
\end{proof}

As a consequence we have the following general deconvolution result which applies to any bounded density whose characteristic function has no zeros, e.g., Gaussians.
\begin{cor}
\label{cor:deconv-phiz}
Assume that $P_Z$ has a density bounded by $m_1$ and the characteristic function $\cf_Z(\omega)$ of $P_Z$ has no zero.
Let 
\begin{equation}
	g(T) = \inf_{|\omega|\leq T} |\cf_Z(\omega)|.
	\label{eq:gTinf}
\end{equation}
Let $P,Q$ have finite first moments and $Q$ has a density $q$ bounded by $m_2$. 
For any $\alpha>0$, let $T(\alpha)$ be the (unique) positive solution to $g(T)^2 =  \alpha T^5$, which satisfies $T(0+)=\infty$.
Then
\begin{equation}
  \dks(P,Q)\leq \frac{C}{T(\TV(P*P_Z,Q*Q_Z))}.
\end{equation}
where $C$ is a constant depending only on $m_1$ and $m_2 + \ExpVal{P}{|X|} + \ExpVal{Q}{|X|}$.

In particular, for $Z\sim \calN(0,1)$,
\begin{equation}\label{eq:deconv_gauss}
  \dks(P,Q)\leq C' \pth{  \log \frac{1}{\TV(P*\calN(0,1),Q*\calN(0,1))} }^{-1/2}.
\end{equation}
where $C'$ is a constant depending only on $m_2 + \ExpVal{P}{|X|} + \ExpVal{Q}{|X|}$.
\end{cor}
\begin{proof}
By assumption, we can choose $g(T)$ in as \prettyref{eq:gTinf} and $h(T)=0$ to fulfill \prettyref{eq:ghIneq}. Then \prettyref{eq:detailTV_KS} leads to
\[
\dks(P,Q) \leq \frac{C}{T} \pth{ 1 +\frac{\sqrt{\TV(P*P_Z,Q*Q_Z) \cdot T^5} }{g(T)}},
\]
where $C_0 = \max\{24 m_2 + 2 (\ExpVal{P}{|X|} + \ExpVal{Q}{|X|}), \sqrt{8 m_1\pi}\}/\pi$.
Since $P_Z$ has a density, $g(T) \leq |\varphi_Z(T)| \to 0$ by Riemann-Lebesgue lemma. 
Since $g(T)$ is decreasing and $g(0)=1$, $\alpha T^5 = g^2(T)$ always has a unique solution $T(\alpha)>0$.
Choosing $T=T(\TV(P*P_Z,Q*Q_Z))$ yields $\dks(P,Q) \leq 2C_0/T$, completing the proof. 
When $Z\sim\calN(0,1)$, we have $g(T)= e^{-T^2/2}$. Choosing $T=\sqrt{-\frac{\log \TV(P*P_Z,Q*P_Z)}{2} }$, the result follows.
\end{proof}

%

\begin{remark} 
Consider a Gaussian $Z$. Then 
$P_n \stackrel{\rm w}{\to} Q \Leftrightarrow P_n * P_Z \stackrel{\rm w}{\to} Q * P_Z \Leftrightarrow P_n * P_Z \stackrel{\rm TV}{\to} P * P_Z$, where the last part follows from pointwise convergence of densities (Scheff\'e's lemma, see, e.g., \cite[1.8.34]{petrov_1995}). 
Furthermore, when one of the distributions has bounded density the Levy-Prokhorov distance (that metrizes
weak convergence) is equivalent to the Kolmogorov-Smirnov distance, cf.~\cite[1.8.32]{petrov_1995}. 
In this perspective, \prettyref{thm:KS_TV}
 can be viewed as a finitary version of the implication 
$\TV(P_n*P_Z,Q*P_Z) \to 0 \Rightarrow \dks(P_n, Q)\to 0$.
\end{remark}

\begin{remark}
A slightly better bound may be obtained if $\ExpVal{P,Q}{ |X+Z|^2} < \infty$. Namely, $T^{3\over 2}$ in the third term
in~\eqref{eq:detailTV_KS} can be reduced to $T$. Indeed if $\delta = \TV(P*P_Z,Q*P_Z)$ then elementary truncation shows
$$ W_1(P*P_Z, Q*P_Z) \lesssim \sqrt{\delta} $$
and then following~\eqref{eq:kantorovich} we get
$$ |\varphi_P(\omega) - \varphi_Q(\omega)| |\varphi_Z(\omega)| \lesssim \sqrt{\delta} |\omega|\,.$$
Now the left-hand side of~\eqref{eq:yp1} can be bounded by $T\over g(T)$ for the choice of $g(T)$ as in~\eqref{eq:gTinf}
and a straightforward modification for the general case of~\eqref{eq:ghIneq}. This improves the constant
in~\eqref{eq:deconv_gauss}.
\end{remark}

\section{Horizontal bound for general additive noise}
\label{sec:GeneralHorizontal}

With the results introduced in the previous section in hand, we are now ready to extend Theorem \ref{thm:capBound} to a broader class of additive noise and channel input distributions.

\begin{thm} 
\label{thm:generalHoriz}
Let $Y=X+Z$ and let $\calP$ be a convex set of distributions. Assume that
\begin{enumerate}[(a)]
\item $P_Z$ satisfies the assumption of \prettyref{lem:DeconvGen};

\item The capacity $C(\calP)\defined \sup_{P_X\in \calP} I(X;Y)$ is finite and attained at some $P_{X^*} \in \calP$.
\end{enumerate}
Then there exists a constant $\epsilon_0$ and a decreasing function $\rho:(0,\epsilon_0)\to (0,\infty)$ (depending on $P_Z$ and $\calP$), such that any $P_{WX}$ with $P_X\in \calP$ satisfies
\begin{equation}\label{eq:genhor}
I(W;X)\geq \rho(C(\calP)-I(W;Y)).
\end{equation} 
Furthermore, if $P_{X^*}$ has no atoms, then $\rho$ satisfies $\rho(0+)=\infty$.
\end{thm}

\begin{remark}
\prettyref{thm:generalHoriz} translates into the following bound on the gap between the $F_I$ curve and the capacity:
\[
F_I(t) \leq C(\calP)- \rho^{-1}(t).
\]
The function $\rho$ can be chosen to be
\begin{equation}
\label{eq:rho}
  \rho(\epsilon)= - \frac{1}{2} \log\left( \calL\left(X^*;T^{-3/4}\right)+\frac{4+2c}{\sqrt{T}} \right),
\end{equation}
where $T=g_1(m_1\sqrt{\epsilon})$, $c$, $g_1, m_1$ are as in \prettyref{lem:DeconvGen}, and 
\begin{equation}
\calL(X^*;\delta)\defined \sup_{x\in \Reals} \Pr[X^*\in B(x,\delta)]
	\label{eq:levy}
\end{equation} is the L\'evy concentration function \cite[p.~22]{petrov_1995} of $X^*$. 
%
For the AWGN channel with $P_Z\sim \calN(0,1)$ and $\calP = \{P_X:\EE{X^2}\leq \gamma\}$ this gives
\begin{align*}
  \rho(\epsilon) 
  &= \frac{1}{8}\log \log \frac{1}{\epsilon} + c_0(\gamma)
\end{align*}
for some constant $c_0(\gamma)$. Compared to the Gaussian-specific bound~\eqref{eq:lnln}, the general proof loses a
factor of two, which is due to the application of Pinsker's inequality.
\end{remark}

%
%
%
%

\begin{proof}
Throughout the proof we assume that  
\begin{equation}
\label{eq:Cepsilon}
C(\calP)-I(W;Y)\leq \epsilon,
\end{equation}
and, from \eqref{eq:IDecomp}, $I(X;Y|W)\leq \epsilon$ and $D(P_X*P_Z\| P_{X^*} * P_Z)\leq \epsilon$, where $P_{X^*}$ is capacity-achieving. Denote
\begin{equation*}
  t(x,w)\defined \TV(P_{Z+x},P_{X|W=w}*P_Z),
\end{equation*}
which is joint measurable in $(x,w)$ for the same reason that $d$ defined in \prettyref{eq:dxw} is jointly measurable. 

Pinsker's inequality yields
\begin{align}
\epsilon &\geq I(X;Y|W) \nonumber \\
        &= \ExpVal{X,W}{D(P_{Z+W}\|P_{X|W}*P_Z)} \nonumber \\
        &\geq 2\Expect[t(X,W)^2] \nonumber \\
        &\geq 2\epsilon \Pr[t(X,W)^2\geq \epsilon]. \label{eq:epsilon-t}
\end{align}
Define
\begin{align*}
  \calF &\defined \{(x,w):t(x,w)\leq\sqrt{\epsilon} \}\\
  \calG &\defined \{w:\exists x,t(x,w)\leq \sqrt{\epsilon}\}.
\end{align*}
Then, from \eqref{eq:epsilon-t},
\begin{equation}
\label{eq:calG_LB}
  \Pr[W\in\calG]\geq \Pr[(X,W)\in \calF]\geq \frac{1}{2}.
\end{equation}
Therefore, for any $w\in \calG$, there exists $\xhw \in \Reals$ such that $t(x,\xhw)\leq \sqrt{\epsilon}$. Applying Lemma \ref{lem:DeconvGen} with $P=P_{X|W=w}$, $Q=\delta_{\xhw}$ and $x_0=T\xhw$, we conclude that
\begin{equation}
  \left|\EE{v(T(X-\xhw))|W=w}-1 \right|\leq \frac{c}{\sqrt{T}},
\label{eq:vDiff1}
\end{equation}
where $v$ is defined in \eqref{eq:vdef}, $c$ is the absolute constant in \eqref{eq:vPQ} and $T=g_1(m_1\sqrt{\epsilon})$.


On the other hand, \eqref{eq:Cepsilon} implies that $D(P_X*P_Z\|P_{Y^*})\leq \epsilon$ and hence $\TV(P_X*P_Z,P_{Y^*})\leq \sqrt{\epsilon}$ by Pinsker's inequality. Applying Lemma \ref{lem:DeconvGen} with $P=P_X$, $Q=P_{X^*}$ and $x_0=T\xhw$, we have
\begin{equation}
\label{eq:vDiff2}
\left|\EE{v(T(X-\xhw))}-\EE{v(T(X^*-\xhw))} \right|\leq \frac{c}{\sqrt{T}}.
\end{equation}
For any $x$, since $v$ takes values in $[0,1]$,
\begin{equation*}
\EE{v(T(X^*-x))}=2\EE{\frac{1-\cos(T(X^*-x))}{T^2(X^*-x)^2}}\leq \Pr[X^*\in B(x,T^{-3/4})]+\frac{4}{\sqrt{T}}.
\end{equation*}
Therefore,
\begin{equation}
\label{eq:vLineq}
0\leq \EE{v(T(X^*-x))}\leq  \calL(X^*;T^{-3/4})+\frac{4}{\sqrt{T}}.
\end{equation}

Using the fact that
\[\TV(P,Q)=\sup_{|f|\leq 1} \int f dP-\int f dQ \] and assembling \eqref{eq:vDiff1}--\eqref{eq:vLineq}, we have for any $w\in \calG$
\begin{align}
\TV(P_X,P_{X|W=w})&\geq \EE{v(T(X-\xhw))|W=w}-\EE{v(T(X-\xhw))} \nonumber\\
                  &\geq 1 - \calL(X^*;T^{-3/4})-\frac{4+2c}{\sqrt{T}}. \label{eq:TV_LB}
\end{align}
Applying \eqref{eq:IWX} and the fact that $D(P\|Q)\geq -\log(1-\TV(P,Q)),$ we have
\begin{align*}
I(W;X)&\geq \EE{\log \frac{1}{1-\TV(P_X,P_{X|W})}}\\
      &\geq \EE{\log \frac{1}{1-\TV(P_X,P_{X|W})}\ones_{W\in\calG}}\\
      &\geq \frac{1}{2} \log\frac{1}{\calL(X^*;T^{-3/4})+\frac{4+2c}{\sqrt{T}}},
\end{align*}
where the last inequality follows from  \eqref{eq:calG_LB} and \eqref{eq:TV_LB}.
\prettyref{lem:levy} in \prettyref{app:levy} implies that $\calL(X^*;0+)=\max_{x \in \reals} \prob{X=x}<1$.
Denote by $\epsilon_0$ the supremum of $\epsilon$ such that $\calL(X^*;T^{-3/4})+\frac{4+2c}{\sqrt{T}} < 1$ and define
$\rho(\epsilon)$ as in \prettyref{eq:rho}. This completes the proof of~\eqref{eq:genhor}. Finally, by
\prettyref{lem:levy} we have that when $P_{X^*}$ is diffuse (i.e. has no atom), it holds that $\rho(0+)=\infty$.
\end{proof}

\section{Infinite-dimensional case}
\label{sec:infinite}

It is possible to extend the results and proof techniques to the case when the channel $X\mapsto Y$ is a $d$-dimensional Gaussian
channel subject to a total-energy constraint
$ \EE{\sum_i X_i^2} \le 1\,.$
Unfortunately, the resulting bound strongly depends on the dimension; in particular, it does not improve the trivial
estimate~\eqref{eq:trivial} as $d\to\infty$. It turns out that this dependence is unavoidable as we show next that \eqref{eq:trivial} holds with equality when $d=\infty$.

To that end we consider an infinite-dimension discrete-time
Gaussian channel. Here the input $X=(X_1,X_2,\dots)$ and $Y=(Y_1,Y_2,\dots)$ are sequences, where
    $Y_i=X_i+Z_i$ and $Z_i\sim \calN(0,1)$ are i.i.d. Similar to Definition~\ref{def:FI}, we define
\begin{equation}
  F^\infty_I(t,\gamma)=\sup \left\{ I(W;Y)\colon I(W;X)\leq t,W\rightarrow
  X \rightarrow Y \right\},
  \label{eq:Finf-trivial}
\end{equation}
where the supremum is over all $P_{WX}$ such that
$\EE{\|X\|^2_2} = \EE{\sum X_i^2} \leq
    \gamma$. Note that, in
this case, 
\begin{equation}
F^\infty_I(t,\gamma)\leq \min\{t,\gamma/2\}.
	\label{eq:trivial-inf}
\end{equation}
The next theorem shows that unlike in the scalar case, there is no improvement over the trivial upper bound \prettyref{eq:trivial-inf} in the infinite-dimensional case.
This is in stark contrast with the strong data processing behavior of total variation in Gaussian noise which turns out to be dimension-free \cite[Corollary 6]{polyanskiy_dissipation_2014}.
\vspace{-0.05in}
\begin{thm}
  \label{thm:infinite}
  $F^\infty_I(t,\gamma) = \min\{t,\gamma/2\}$.
\end{thm}
\begin{proof}
	For any $\epsilon>0$ and all sufficiently large $\beta > 0$, there exists $n$ and a code of size of $M_\beta$ for the $n$-parallel Gaussian channel, where each codeword has energy (squared $\ell_2$-norm) less than $\beta$, the probability of error is at most $\epsilon$, and $M_\beta = e^{\beta/2+o(\beta)}$ as $\beta\to \infty$ (see, e.g.\cite[Thm. 7.5.2]{gallager_information_1968}).
Choosing $X$ uniformly at random over the codewords, we have
  from Fano's inequality
  \begin{align*}
    I(X;Y)\geq (1-\epsilon)\log M - h(\epsilon) = \frac{(1-\epsilon)\beta}{2}
    +o(\beta)-h(\epsilon).
  \end{align*}
 For any $\beta > \gamma$, define
  \begin{align*}
    X'= 
        \begin{cases}
            x_0 & \mbox{w.p. } 1-\frac{\gamma}{\beta}\\
            X & \mbox{w.p. } \frac{\gamma}{\beta}.
        \end{cases}  
  \end{align*}
  where $x_0$ is an arbitrary vector outside the codebook.
  Then, $\Expect[\|X'\|_2^2] \leq \gamma$. Furthermore, as $\beta \diverge$,
  \begin{align*}
H(X') = \frac{\gamma}{\beta}\log M +h\left( \frac{\gamma}{\beta} \right) = \frac{\gamma}{2} + o(1),
  \end{align*}
  and, by the concavity of the mutual information in the input distribution, 
  \begin{align*}
I(X';Y) \geq \frac{\gamma}{\beta} I(X;Y) \geq \frac{(1-\epsilon)\gamma}{2} + o(1).
 \end{align*}
 Since $F_I^\infty(\gamma/2,\gamma) \geq \frac{I(X';Y)}{H(X')}$, first sending $\beta \diverge$ then $\epsilon \to 0$, we have 
    $F_I^\infty(\gamma/2,\gamma)=\gamma/2$. The
    result then follows by noting that $t \mapsto F_I^\infty(t,\gamma)/t$ is decreasing
    and $t \mapsto F_I^\infty(t,\gamma)$ is increasing (\prettyref{prop:FIprop}).
\end{proof}

\appendices

\section{Alternative version of Lemma \ref{lem:capKS}}

\begin{lem}
\label{lem:KStalagrand}
Assume that $C(\gamma)-I(X;Y_\gamma)\leq \epsilon < 1$. 
  Then  
    \begin{align}
      d_{\mathrm{KS}}\left(P_{X},\calN(0,1)\right) \leq
        \frac{24}{\pi^{3/2}\sqrt{\gamma\log(1/\epsilon)}}+\frac{2\sqrt{2(1+\gamma)} \epsilon^{1/4} \sqrt{\log
    (1/\epsilon )}}{\pi }
    \end{align}
\end{lem}
\begin{proof}
Abbreviate $Y_\gamma = \sqrt{\gamma} X + Z$ by $Y$.
  From Talagrand's inequality \cite[Thm 1.1]{talagrand_transportation_1996}
    \begin{equation*}
        W_2(P_{\sqrt{\gamma}X}*\calN(0,1),\calN(0,\gamma+1))\leq 2\sqrt{(1+\gamma)\epsilon},
    \end{equation*}
where $W_p$ is the Wasserstein distance, given by $$W_p(P,Q)\defined\inf_{P_{X,Y}:X\sim P,Y\sim Q}\EE{\|X-Y\|_p}. $$
Since $W_1(\mu,\nu)\leq W_2(\mu,\nu)$ for any measures $\mu,\nu$, 
there         exists a random variable $G\sim \calN(0,\gamma+1)$ such that
        \begin{equation}
        \label{eq:W1bound}
            \EE{|Y-G|}\leq 2\sqrt{(1+\gamma)\epsilon}.
        \end{equation}

Let $\varphi_Y(t)$ and $\varphi_G(t)$ be the characteristic functions of $Y$ and $G$, respectively. Then
    \begin{align}
|\varphi_Y(t)-\varphi_G(t)|&=\left| \EE{e^{itY}-e^{itG}} \right|
                                   \leq \EE{|t(Y-G)|}
                                   \leq 2|t|\sqrt{(1+\gamma)\epsilon}\label{eq:kantorovich}
    \end{align}
where the second inequality follows from  \cite[Lemma
4.1]{feller_introduction_1966}, and the last inequality from \eqref{eq:W1bound}.
Using Esseen's inequality (\prettyref{lem:Esseen}) and the fact that the
PDF of $G$ is upper bounded by $1/\sqrt{2\pi P}$, for all $T>0$
    \begin{align*}
        \left|P_{\sqrt{\gamma}X}(t)- \calN(0,P)\right|\leq& \frac{1}{\pi} \int_{-T}^T
        \left|\frac{\varphi_X(t)-e^{-\gamma t^2/2} }{t}\right| dt + \frac{12\sqrt{2}}{\pi^{3/2}T\sqrt{\gamma}}\\
                                       =& \frac{1}{\pi}\int_{-T}^T e^{t^2/2}
                                       \left|\frac{\varphi_Y(t)-\varphi_G(t)
                                       }{t}\right|dt + \frac{12\sqrt{2}}{\pi^{3/2}T\sqrt{\gamma}}\\
                                       \leq& \frac{4\sqrt{(1+\gamma)\epsilon}T
                                       e^{T^2/2}}{\pi}+\frac{12\sqrt{2}}{\pi^{3/2}T\sqrt{\gamma}}.
     \end{align*}
    Choosing $T=\sqrt{\frac{1}{2}\log(1/\epsilon)}$ yields 
\begin{equation}
  \left|P_{\sqrt{\gamma} X}(t)- \calN(0,\gamma)\right|\leq \frac{2\sqrt{2(1+\gamma)} \epsilon^{1/4} \sqrt{-\log
    (\epsilon )}}{\pi } +\frac{24}{\pi ^{3/2}  \sqrt{-\gamma\log (\epsilon )}}.
\end{equation}
The proof is complete upon observing that 
    $\dks(P_{\sqrt{\gamma}X},\calN(0,\gamma))=\dks\left(P_{X},\calN(0,1)\right)$.
\end{proof}

\section{L\'evy concentration function near zero}
	\label{app:levy}
We show that the L\'evy concentration function defined in \prettyref{eq:levy} is continuous at zero if and only if the distribution has no atoms.
This fact is used in the proof of \prettyref{thm:generalHoriz}.

\begin{lem}
	For any $X$, $\lim_{\delta \to 0} \calL(X;\delta) = \max_{x \in \reals} \prob{X=x}$. Consequently,
	$\calL(X;0+)=0$ if and only if $X$ has no atoms.
	\label{lem:levy}
\end{lem}
\begin{proof}
Let $a \triangleq \lim_{\delta \to 0} \calL(X;\delta)$, which exists since $\delta \mapsto \calL(X;\delta)$ is increasing.
	Since $\calL(X;\delta) \geq \prob{X=x}$ for any $\delta>0$ and any $x$, it is sufficient to show that  $a \leq \max_{x \in \reals} \prob{X=x}$. Assume that $a>0$ for otherwise there is nothing to prove.
	By definition, for any $n$, there exists $x_n$ so that $\prob{X\in B(x_n,1/n)} \geq a - 1/n$. 	
	Let $T>0$ so that $\prob{|X|>T} \leq a/2$. Then $|x_n| \leq T$ for all sufficiently large $n$. By restricting to a subsequence, we can assume that $x_n$ converges to some $x$ in $[-T,T]$. By triangle inequality, $\prob{X \in B(x, |x_n-x|+1/n)} \geq \prob{X \in B(x_n, 1/n)} \geq a-1/n$. 
	By bounded convergence theorem, $\prob{X=x} \geq a$, completing the proof.
\end{proof}

\bibliographystyle{alpha}
\bibliography{IEEEabrv,references}

\end{document}